\documentclass[sigconf, nonacm]{acmart}

\newcommand\vldbdoi{XX.XX/XXX.XX}

\newcommand\vldbavailabilityurl{URL_TO_YOUR_ARTIFACTS}
 
\usepackage{subcaption}
\usepackage{standalone}
\usepackage{color, colortbl}
\usepackage{listings}

\usepackage[vlined,commentsnumbered,ruled]{algorithm2e}
\let\oldnl\nl
\newcommand{\nonl}{\renewcommand{\nl}{\let\nl\oldnl}}

\SetCommentSty{mycommfont}

\def\HiLi{\leavevmode\rlap{\hbox to \hsize{\color{red!20}\leaders\hrule height .8\baselineskip depth .5ex\hfill}}}

\newcommand\independent{\protect\mathpalette{\protect\independenT}{\perp}}
\def\independenT#1#2{\mathrel{\rlap{$#1#2$}\mkern2mu{#1#2}}}

\usepackage{amsfonts}
\usepackage{amsmath}
\usepackage{multirow}
\usepackage{paralist}
\usepackage{paralist}
\usepackage{enumitem}

\newtheorem{definition}{Definition}[section]
\newtheorem{theorem}{Theorem}[section]

\newtheorem{example}{Example}[section]

\newcommand{\revised}[1]{\textcolor{black}{~#1}{\typeout{#1}}}

\newcommand{\sysName}{\textsc{MESA}}
\newcommand{\probName}{\textsc{Correlation-Explanation}}

\begin{document}

\setlength{\abovedisplayskip}{0pt}
\setlength{\belowdisplayskip}{0pt}
\setlength{\abovedisplayshortskip}{0pt}
\setlength{\belowdisplayshortskip}{0pt}

\title{On Explaining Confounding Bias}


\author{Brit Youngmann}
\affiliation{%
  \institution{CSAIL MIT}
  \country{}
}
\email{brity@mit.edu}

\author{Michael Cafarella}
\affiliation{%
  \institution{CSAIL MIT}
    \country{}
}
\email{michjc@csail.mit.edu}

\author{Babak Salimi}
\affiliation{%
  \institution{University of California, San Diego}
    \country{}
}
\email{bsalimi@ucsd.edu}

\author{Yuval Moskovitch}
\affiliation{%
  \institution{Ben Gurion University of the Negev}
    \country{}
}
\email{yuvalmos@bgu.ac.il}

\begin{abstract}
When analyzing large datasets, analysts are often interested in the explanations for surprising or unexpected results produced by their queries. In this work, we focus on aggregate SQL queries that expose correlations in the data. A major challenge that hinders the interpretation of such queries is \emph{confounding bias}, which can lead to an unexpected correlation. We generate explanations in terms of a set of \emph{confounding variables} that explain the unexpected correlation observed in a query. We propose to mine candidate confounding variables from external sources since, in many real-life scenarios, the explanations are not solely contained in the input data. We present an efficient algorithm that finds the optimal subset of attributes (mined from external sources and the input dataset) that explain the unexpected correlation. This algorithm is embodied in a system called \sysName. We demonstrate experimentally over multiple real-life datasets and through a user study that our approach generates insightful explanations, outperforming existing methods that search for explanations only in the input data. We further demonstrate the robustness of our system to missing data and the ability of \sysName\ to handle input datasets containing millions of tuples and an extensive search space of candidate confounding attributes.



\end{abstract}
\maketitle

\settopmatter{printacmref=false}


\section{Introduction}



When analyzing large datasets, analysts often query their data to extract insights. Oftentimes, there is a need to elaborate upon the queries' answers with additional information to assist analysts in understanding unexpected results, especially for aggregate queries, which are harder to interpret~\cite{salimi2018bias,lin2021detecting}.
While aggregate query results expose correlations in the data, the human mind cannot avoid a causal interpretation. Thus, we provide explanations for unexpected correlations observed in aggregate queries using causation terms.

In this work, we focus on SQL queries that are aggregating an {\em outcome attribute} ($O$) based on some groups of interest indicated by a grouping attribute, referred to as the {\em exposure} ($T$) \cite{pearl2009causality}. 
A major challenge that hinders the interpretation of such queries is {\em confounding bias}~\cite{pourhoseingholi2012control} that can lead to a spurious association between $T$ and $O$ and hence perplexing conclusions.
Confounding bias occurs when an analyst tries to determine the effect of an exposure on an outcome
but unintentionally measures the effect of another factor(s) (i.e., a
\emph{confounding variable(s)}) on the outcome. This results in a distortion of the actual association between $T$ and $O$~\cite{pearl2009causality}. We are interested in generating explanations in terms of a set of confounding variables that explain unexpected correlations observed in query results.

Previous work detected uncontrolled confounding variables from the data~\cite{salimi2018bias}. However, in many cases, such variables might be found outside the narrow query results that and the database being used~\cite{li2021putting}. Thus, there is a need to develop automated solutions that can explain unexpected correlations observed in query results to analysts, which goes beyond just the data accessed by the query. To illustrate, consider the following example.

\begin{example}
\label{ex:intro_ex1}
Ann is a data analyst in the WHO organization who aims to understand the coronavirus pandemic for improved policymaking.
She examines a dataset containing information describing Covid-19-related facts in multiple cities worldwide.
It consists of the number of deaths-/recovered-/active-/new- per-100-cases in each city. 
Ann evaluates the following query over this dataset: 
\begin{center}
\small
    \begin{tabular}{l}
         \verb"SELECT Country, avg(Deaths_per_100_cases)"\\
         \verb"FROM Covid-Data"\\
         \verb"GROUP BY Country"\\
    \end{tabular}
\end{center}
A visualization of the query results is given in Figure \ref{fig:covid_example}.
Here, the exposure is \textsc{Country} and the outcome is \textsc{Deaths\_per\_100\_cases}.
Ann observes a puzzling correlation between the exposure and outcome; namely, she wonders why the choice of the country has such a substantial effect on the death rate. 
She is interested in finding a set of confounding variables that explain this association. 
She sees that the attribute \textsc{Confirmed\_cases} from \textsc{Covid-data} is correlated with \textsc{Deaths\_per\_100\_cases}. However, this attribute alone is not enough to explain the correlation. 
For example, she sees that while Germany had the fifth-most confirmed cases worldwide, it had only a fraction of the death toll in other countries.
Ann understands that other factors (that are not in this data) affect the association between death rate and country. 
She remembers reading in the news that as a country’s success (defined by multiple variables, including GDP\footnote{Gross domestic product (GDP) is the monetary value of all goods and services made within a country during a specific period.} and HDI\footnote{The Human Development Index (HDI) is a statistic composite index of life expectancy, education, and per capita income indicators.}) grows, the death rate decreases~\cite{upadhyay2021correlation,kaklauskas2022effects}. 
However, such economic features of countries are not available in her dataset but could be extracted from external sources. 
\end{example}

\begin{figure}[]
\begin{center}
		\includegraphics[scale = 0.4]{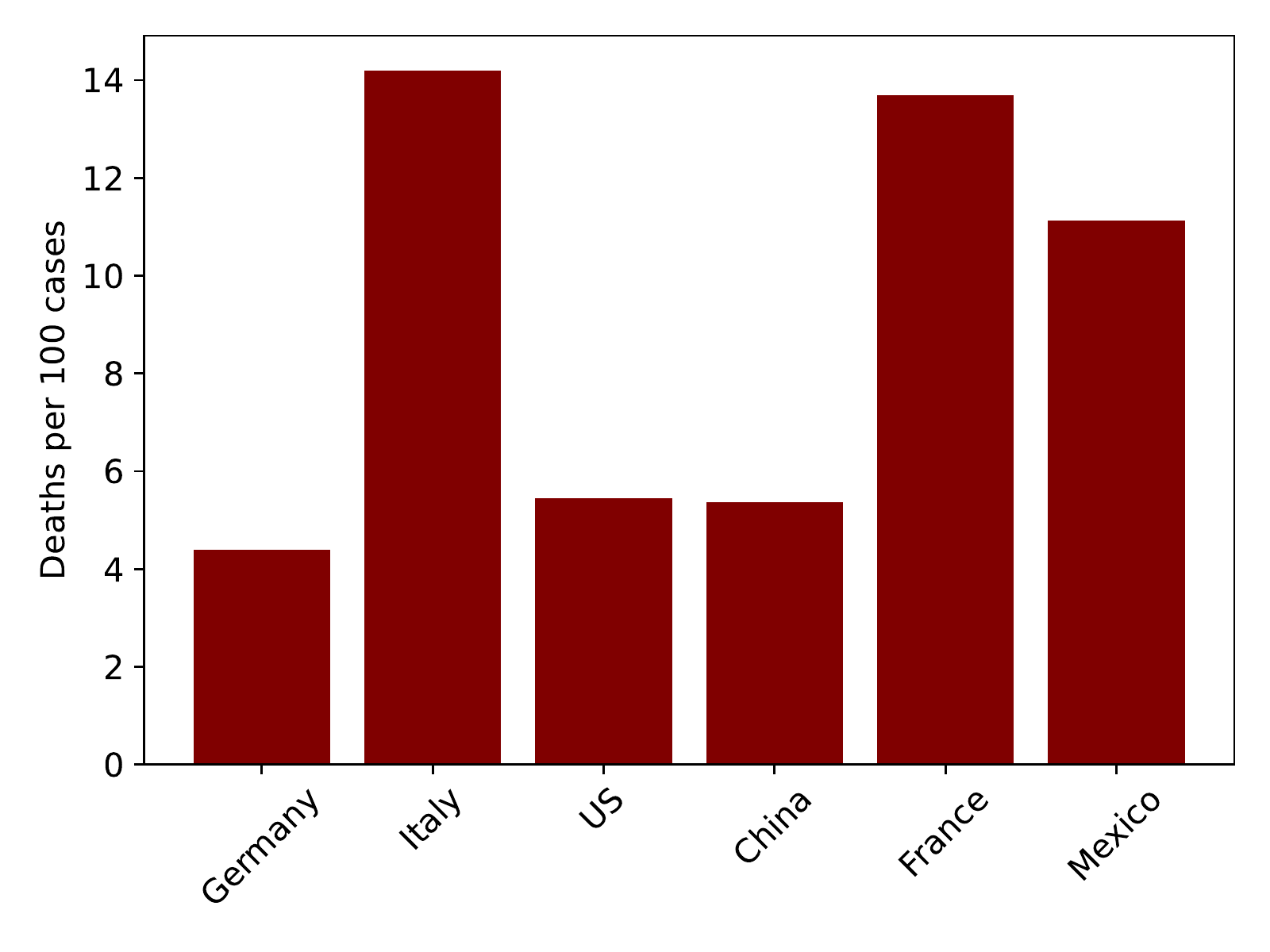}
		\vspace{-8px}
	
		\caption{Visualization of the results of the query $Q$. } 
			\label{fig:covid_example}
\end{center}
\vspace{-6mm}
\end{figure}

We propose to mine unobserved confounding attributes from external sources. In general, our framework can extract candidate confounders from any knowledge source (e.g., related tables, data lakes, web tables) as long as it can be integrated with the input data. This paper focuses on mining attributes from a Knowledge Graph (KG) for the following reasons. KGs are an emerging type of knowledge representation \cite{akrami2020realistic,zheng2018question,de2016deepdive}, that can effectively organize and represent a large amount of data.
KGs have been efficiently utilized in various tasks, such as question-answering and recommendation~\cite{chen2020review}. Further, attribute names in KGs are typically highly informative, allowing analysts to reason about the generated explanations. 
However, the sheer breadth of coverage that makes KGs potentially valuable also creates the need to automate the process of mining relevant confounding variables. 
There are multiple general-purpose (e.g., Wikidata~\cite{wikidata}, DBpedia~\cite{dbpedia}, Yago~\cite{rebele2016yago}) or domain-specific (e.g., for medical proteomics~\cite{santos2022knowledge}, or protein discovery~\cite{mohamed2020discovering}) KGs that act as central storage for data collected from multiple sources. We argue that such valuable data could be utilized for explaining unexpected correlations observed in user queries in a wide range of scenarios.

To this end, we present an efficient algorithm that finds a subset of confounding attributes (mined from external sources and the input dataset) that explain the unexpected correlation observed in a given query. 
This algorithm is embodied in a system called \sysName, which automatically mines candidate attributes from a knowledge source. This source may be provided by the analyst (for a specific domain) or could be any publicly available knowledge source. 

\begin{example}
Ann uses \sysName\ to search for an explanation for her query. \sysName\ mines all available attributes about countries that appear in her data from DBpedia. She learns that besides \textsc{Confirmed\_cases}, the attributes \textsc{HDI}, and \textsc{GDP} are uncontrolled confounding attributes. She sees that the death rate is similar in countries with a similar number of confirmed cases, HDI, and GDP. 
She is pleased because she found a plausible real-world explanation for her query results~\cite{upadhyay2021correlation,kaklauskas2022effects}. 
\end{example}

Previous work provides explanations for trends
and anomalies in query results in terms of predicates on attributes that are shared by one (group of) tuple in the results but not by another (group of) tuple \cite{el2014interpretable,li2021putting,roy2015explaining,roy2014formal,wu2013scorpion}. However, those methods do not account for correlations among attributes, and are thus inapplicable for explaining the correlation between the outcome and the exposure.
\cite{salimi2018bias} presented a system that provides explanations based on causal analysis, measured by correlation among attributes. However, this system only considers the input dataset, and its running times are exponential in the number of candidate confounding attributes. 
We share with CajaDE~\cite{li2021putting} the motivation of considering explanations that are not solely
drawn from the input table. CajaDE is a system that generates insightful explanations based on contextual
information mined from tables related to the table
accessed by the query. Their explanations are a set of patterns that are unevenly distributed among groups in the query results, and are \emph{independent of the outcome attribute}. Thus, CajaDE may generate explanations 
that are irrelevant for understanding the correlation between the exposure and outcome. 

Our framework supports a rich class of aggregate SQL queries that compare among subgroups, investigating the relationship between an aggregated attribute $O$ and a grouping attribute $T$.  
To explain the correlation between $T$ and $O$ observed in the results of a query $Q$, we formalize the \probName\ problem that seeks a set of confounding attributes (extracted from external sources or the input database), which minimizes the partial correlation between $T$ and $O$ (to measure the correlation between $T$ and $O$, while controlling for the effect of confounding variables).
Further, 
\sysName\ enables analysts to learn the individual responsibility of selected attributes and to automatically identify unexplained data
subgroups (correspond to refinements
of $Q$) in which different explanations are required.


Given an input database $\mathcal{D}$ and a knowledge source, we extract a set of attributes representing additional properties of entities from $\mathcal{D}$.
The attributes are extracted only after the query arrives (as the knowledge source may be a part of the input). Extracted attributes may contain many missing values, especially ones extracted from a KG where data is sparse. 
Previous work showed that common approaches for handling missing data could cause substantial \emph{selection bias}~\cite{seaman2013review} (which occurs when the obtained data fails to properly represent the population intended to be analyzed) if many values are missing~\cite{seaman2013review}. In contrast to prediction, explanations quality is more sensitive to missing data~\cite{mohan2018estimation}. 
We, therefore, present a principled way of handling missing values, ensuring the explanations are robust to missing data. We provide sufficient conditions to detect selection bias and an algorithmic approach to handle it properly.

There are potentially hundreds of attributes that could be extracted from external sources. Thus, there is a need to develop an efficient algorithm to search for the optimal attribute set (i.e., explanation) in this extensive search space. Further, the search for the optimal attribute set involves estimating partial correlation
for high-dimensional conditioning sets, which is notoriously difficult~\cite{li2017feature}. 
To this end, we propose the MCIMR algorithm, which does not require iterating over all possible attribute sets, and avoids estimating high-dimensional
conditioning sets. 
It selects attributes based on Min-Conditional-mutual-Information (a common measure for partial correlation) and Min-Redundancy criteria, 
yielding a PTIME algorithm that finds the optimal $k$-size explanation where $k$ is given. 
We then define a stopping criterion, allowing the algorithm to stop when no further improvement is found. We propose multiple pruning techniques to speed up the computation.

We conduct an experimental study based on four commonly used datasets that evaluate the quality and efficiency of the MCIMR algorithm.
Our approach is effective whenever the explanation can be found in a given knowledge source. 
We show that this was the case in $72.5\%$ of random aggregate queries evaluated on these datasets, setting the knowledge source to be the DBPedia KG~\cite{dbpedia}. 
To evaluate the explanations quality, we focus on $14$ representative queries suffering from confounding bias. These queries
are inspired by real-life analysis reports, such as Stack Overflow
annual reports~\cite{stackoverflowreport} and academic papers~\cite{upadhyay2021correlation}.
We ran a user study consisting of $150$ subjects to evaluate the quality of our explanations compared with six approaches. We show that the explanations generated by MCIMR are almost as good as those of a computationally infeasible na\"ive method that iterates over all attribute subsets and are much better than those of feasible competitors. We also show that previous findings in each domain support our substantive explanations.
Our experimental results demonstrate the robustness of our solution to missing data and indicate the effectiveness of our algorithm in finding explanations in less than $10$s for queries evaluated on datasets containing more than $5$M tuples.

Our main contributions are summarized as follows:
\begin{itemize}[noitemsep,topsep=0pt,parsep=0pt,partopsep=0pt]
\item We formalize the \probName\ problem that seeks a subset of attributes that explains unexpected correlations observed in SQL queries (Section \ref{sec:model}). 
\item We propose to extract unobserved confounding attributes from external sources and focus on KGs. We develop a principled way to avoid selection bias (Section \ref{sec:attribute_extraction}). 
\item We devise an efficient algorithm that computes the optimal explanation for the \probName\ problem. We embody this algorithm in a system called \sysName\, which enables analysts to automatically identify unexplained data
subgroups (Section \ref{sec:MCIMR}).
\item We qualitatively
evaluated the explanations produced by \sysName\ and existing solutions over real-life datasets through a user study. We further conducted performance experiments to assess scalability 
(Section \ref{sec:experiments}). 
\end{itemize}

\vspace{1mm}
Related work is presented in Section~\ref{sec:related} and we conclude
in Section~\ref{sec:conc}. 
 \section{Model and Problem Formulation}\label{sec:model}

\subsection{Data Model}
\label{subsec:model}
\revised{We operate on a standard multi-relational dataset $\mathcal{D}$. To simplify the exposition, we assume $\mathcal{D}$ consists of
a single relational table, however, our definitions and results apply to the general case.}
The table's attributes are denoted by $\mathcal{A}$. For an attribute $A_i$ we denote its domain by $Dom(A_i)$. We use 
bold letters
for sets of attributes $\boldsymbol{A} {\subseteq} \mathcal{A}$. 
We expect the reader is familiar with basic information theory measures, such as entropy and conditional mutual information.  

\revised{Our framework supports a rich class of SQL queries that
involve groping, joins and different aggregations to support complex real-world scenarios. The queries we examine compare among subgroups, investigating the relationship between an aggregated attribute $O$ (referred to as the outcome) and an grouping attribute $T$ (referred to as the exposure).} To simplify the exposition, we assume a single grouping attribute. However, our results can be naturally generalized for multiple grouping attributes. To handle a numerical exposure, one may bin this attribute.
We call the condition
$C$ (given by the WHERE clause) the context for the query. 
Given $C$, we aim to explain the difference
among $agg(O
)$ for each $T {=} t_i$, where $t_i {\in} Dom(T)$. \revised{ If the attribute $O$ belongs to a different table from the one containing the exposure $T$, the query $Q$ describes how these two tables
are combined in the join condition. }




We use the following example based on the Stack Overflow (SO) dataset throughout this paper. In our experiments, we demonstrate the operation of \sysName\ over four datasets, including Covid-Data. 
\begin{example}
\label{ex:running_ex}
SO dataset contains information about people who code around the world, such as their
age, gender, income, and country. 
Consider the following query: 
\begin{center}
    \begin{tabular}{l}
\verb"SELECT Country, avg(Salary)"\\
\verb"FROM SO"\\
\verb"WHERE Continent = Europe"\\
\verb"GROUP BY Country"\\
    \end{tabular}
\end{center}
Here, $O$ is \textsc{Salary}, $T$ is \textsc{Country}, the context $C$ is \textsc{Continent = Europe}, and the aggregation function is average. We aim to explain the difference
in the average salary of developers from each country in Europe. While some attributes from the dataset may partially explain this (e.g., \textsc{Gender, DevType}), other important attributes that can cast light on this difference cannot be found in this dataset. 
\end{example}

\textit{Knowledge Extraction}. 
\revised{In general, \sysName\ can extract attributes from any external source, such as related tables, data lakes, unstructured data (e.g., images), or Knowledge Graphs (KGs), as long as it can be integrated with the input dataset. 
This paper focuses on mining attributes from a KG for the following reasons. First, KGs can effectively organize a large amount of (domain specific or general) data, and have been successfully utilized in various downstream applications, such as question-answering systems,
search engines, and recommendation systems~\cite{chen2020review}. Second, one of the strengths of KGs is that most of the attributes are already reconciled. Namely, we will not have to match different versions of attributes across different entities. Last, the attribute names are typically highly informative, allowing users to reason about the generated explanations. We note that extracting attributes from other sources poses a series of additional challenges, including handling many-to-many relations and uninformative attribute names. We leave these extensions for future research.}

Attributes extracted from a knowledge source may be irrelevant for a given query. 
We thus let the analyst decide which source \sysName\ should use.
Given a knowledge source (e.g., domain specific KG~\cite{santos2022knowledge,mohamed2020discovering}, publicly available KG~\cite{wikidata,dbpedia,rebele2016yago}, data lake), we extract a set of attributes $\mathcal{E}$ representing additional properties of entities from $\mathcal{D}$. 


Continuing with our example, $\mathcal{E}$ could be a set of properties of countries extracted from a KG, such as their density, and HDI.
We can potentially join $\mathcal{E}$ and $\mathcal{T}$, by linking values from $\mathcal{T}$ with their corresponding entities in $\mathcal{G}$ that were used for attributes extraction. However, $\mathcal{E}$ may contain many attributes, most of them are irrelevant for explaining the query results. 


 \subsection{Problem Formulation}
\label{sucsec:problem}

Given a query, the analyst observes an unexpected correlation between the exposure $T$ and the outcome $O$ attributes that she would like to explain.
We assume there is \emph{confounding bias} that causes a spurious association
between $T$ and $O$.  
Confounding bias is a systematic error due to
the uneven or unbalanced distribution of a third variable(s), known
as the confounding variable(s) in the competing groups. Uncontrolled confounding variables 
lead to an inaccurate estimate of the true association between $T$ and $O$.
Our goal is to discover the confounding variables. 
Let $\mathcal{A}$ denote $\mathcal{E}{\cup} \mathcal{T}{\setminus}\{O,T\}$, referred to as the candidate attributes. 
$\mathcal{A}$ contains confounding attributes that affect both $T$ and $O$. We aim to find an attribute set $\boldsymbol{E}{\subseteq} \mathcal{A}$ that control the correlation between $O$ and $T$, i.e., when conditioning on $\boldsymbol{E}$, the correlation between $O$ and $T$ is diminished.
We call such a set the correlation explanation. 

\begin{example}
\label{ex:problem_ex}
It is very likely that countries' economic features (such as GDP, Gini, and
HDI) affect developers' salaries. To unearth the association between \textsc{Country} and \textsc{Salary}, one must measure the correlation while controlling for such attributes. This will allow users to understand which factors affect the differences in developers' salaries. Intuitively, we expect the average developers' salaries to be similar in countries with similar economic characteristics. 
\end{example}

Ideally, we look for a minimal-size set of attributes $\boldsymbol{E} {\subseteq} \mathcal{A}$ s.t: $(O {\independent} T | \boldsymbol{E}, C)$. However, in practice, we may not find such perfect explanations (that entirely explains away the correlation), hence we search for a minimal-size set of attributes that \emph{minimizes the partial correlation between $T$ and $O$}. 
Partial correlation measures the strength of a relationship between two variables, while controlling for the effect of other variables. A common measure of partial correlation is multiple linear regression, which is sensitive only to linear relationship.
Other partial correlation measures, such as Spearman's coefficient, are more sensitive to nonlinear relationships \cite{croxton1939applied, esmailoghli2021cocoa}.
Here we use \emph{Conditional Mutual Information} (CMI),
a common measure of the mutual dependence between two variables, given the value of a third. 
We chose CMI because
(1) it is a widely used non-parametric measure for partial correlation~\cite{chandrashekar2014survey}, (2) there is a plethora of techniques for estimating it from data~\cite{salimi2018bias}, (3) it also allows us to develop information-theoretic optimizations. CMI may suffer from underestimation, especially when quantifying dependencies among variables with high associations~\cite{zhao2016part}. However, we avoid such cases since, as we explain in Section \ref{subsec:optimizations}, we discard all attributes that are logically dependent on $T$ or $O$.
Note that $(O {\independent} T | \boldsymbol{E}, C)$ holds iff $I(O;T | \boldsymbol{E}, C) {=} 0$, where $I(O;T|\boldsymbol{E},C)$ is the mutual information of $O$ and $T$ while conditioning on $\boldsymbol{E}$ and the context $C$. 
Thus, we formalize the \probName\ problem as follows: 

\begin{definition}[\textbf{\probName}]
\label{problem:problem1}
Given a set of candidate attributes $\mathcal{A}$ and a query $Q$, find a set of attributes $\boldsymbol{E}^*$ s.t.: $\boldsymbol{E}^* = argmin_{\boldsymbol{E} {\subseteq} \mathcal{A}} I(O;T|\boldsymbol{E},C){\cdot}|\boldsymbol{E}|$. 
\end{definition}

Following previous work~\cite{roy2014formal,pradhan2021interpretable,DBLP:journals/pvldb/LockhartPWW021}, besides the explanatory power, we also consider the cardinality of the sets. 



\begin{example}
\label{ex:running_ex3}
Among other attributes, we extracted from a KG the \textsc{Gini} ($E_1$), \textsc{Density} ($E_2$), and \textsc{HDI} ($E_3$) attributes. An attribute from SO is the developers \textsc{Gender} ($E_4$).
According to our data, we have $I(O;T|C) {=} 2.6$. When conditioning on $E_1$, we get: $I(O;T|C, E_1) {=} 1.3$. Namely, 
in countries with a similar Gini index, there is less correlation between the country of developers and their salaries.
When also considering \textsc{Density}, we get: $I(O;T|C, E_1, E_2) {=} 0.03$. Thus, 
this set of attributes explains away the correlation in $Q_{so}$.
When conditioning on \textsc{HDI}, on the other hand, we get: $I(O;T|C, E_3) {=} 2.5$. Since the HDI of all countries in Europe is similar\footnote{As reflected in \url{https://en.populationdata.net/rankings/hdi/europe/}.}, this attribute does not explain the observed correlation.
Similarly, when conditioning on \textsc{Gender} we get: $I(O;T|C, E_4) {=} 2.3$, implying that the developers gender cannot explain the correlation in $Q_{so}$.
\end{example}

To
assist analysts in interpreting the
results, we enable them to learn the individual \emph{responsibility} of selected attributes.
Given an explanation $\boldsymbol{E}$, we
rank the attributes in $\boldsymbol{E}$ in terms of
their responsibilities as follows:
\begin{definition}[Degree of responsibility]
\label{def:resp}
Given a query $Q$ and set of attributes $\boldsymbol{E}$, the degree of responsibility of an attribute $E_i {\in} \boldsymbol{E}$ is defined as follows:
$$Resp(E_i) := \frac{I(O;T|\boldsymbol{E} {\setminus} \{E_i\}, C)- I(O;T|\boldsymbol{E}, C)}{\sum_{E_j \in \boldsymbol{E}}(I(O;T|\boldsymbol{E} {\setminus} \{E_j\}, C) -I(O;T|\boldsymbol{E},C))}$$
\end{definition}
The responsibility of an attribute $E_i$ is the normalized value of its individual contribution. 
When all attributes in $\boldsymbol{E}$ contribute to the explanations (i.e., the numerator is positive), the denominator is non-negative. 
The responsibility of $E_i$ is positive if $E_i$ contributes to the explanation. 
Thus, a negative responsibility indicates that $E_i$ only harms the explanation (it happens since $E_i$ has a negative interaction information with $O$ and $T$). The higher the responsibility of an attribute, the greater is its individual explanation power.

\begin{example}
\label{ex:resposibility}
Recall that $E_1 {=} $ \textsc{Gini}, and $E_2{=}$ \textsc{Density}. Let $\boldsymbol{E} {=} \{E_1, E_2\}$. 
According to our data we have: $I(O;T|C, E_2) {=} 1.51$. We get:
$Resp(E_1) {=} 0.54$, and
$Resp(E_2) {=} 0.46
$.
The attribute \textsc{Hobby} ($E_5$) indicates whether a developer is coding as an hobby. It has a negative interaction information with $O$ and $T$. We have $I(O;T|C, E_5) {=} 2.7$ ${>} I(O;T|C)$. Let $\boldsymbol{E} {=} \{E_1, E_5\}$. We get: $I(O;T|C,\boldsymbol{E}) {=} 1.5$, $
    Resp(E_1) {=} 1.2$, and
$Resp(E_5) {=} -0.2
$. 
Since $E_5$ did not contribute to the explanation, its responsibility is negative. 
\end{example}

\textit{Key Assumption}.
We generally believe that attributes with low responsibility are of little interest to analysts and that XOR-like explanations (in which the explanation power of each individual attribute is low, but their combination makes a good explanation) are hard to understand; thus, they are less likely to be considered good explanations. Our view is motivated by~\cite{lombrozo2007simplicity}. A similar assumption is often made in feature selection \cite{tsamardinos2003algorithms,brown2008markov}, where they assume the optimal feature set does not contain multivariate associations among features, which are individually irrelevant to a target class but become
relevant in the presence of others. We further believe true XOR phenomena are likely to be uncommon in real datasets; the practical success of feature selection methods that make this assumption~\cite{chandrashekar2014survey} is some evidence for this view. 
Further, generating XOR explanations would be a substantial additional technical challenge. It would eliminate our ability to prune low-relevance attributes and to define a stopping criterion for our algorithm (see Section \ref{sec:MCIMR}). Also, extending our algorithm to consider XOR explanations would mean estimating CMI for a high-dimensional conditioning set, which is notoriously difficult~\cite{li2017feature}.



\section{Attributes Extraction}
\label{sec:attribute_extraction}

\subsection{Extracting the Candidate Attributes}
\label{subsec:extraction}
\revised{\sysName\ extracts attributes representing additional properties of entities from $\mathcal{D}$ from a given knowledge source. 
In general, we may extract attributes from any given source as long as it can be integrated with the input dataset. For example, we may extract attributes from a data lake, leveraging existing methods to join or union an input table with other tables~\cite{zhang2020finding, nargesian2018table,zhu2019josie,esmailoghli2021cocoa,santos2021correlation}.
As mentioned in Section \ref{subsec:model}, here we focus on extracting attributes from a given KG.}

\textbf{Extracting Attributes from a KG}: 
Given a KG, the first step is to map values that appear in the table $\mathcal{T}$ to their corresponding unique entities in the KG $\mathcal{G}$. This task is often referred to as the Named Entity Disambiguation (NED) problem~\cite{parravicini2019fast}. We can use any off-the-shelf NED algorithm (e.g., \cite{parravicini2019fast,zhu2018exploiting}) to match any non-numerical value in $\mathcal{T}$ to an entity in $\mathcal{G}$. 
Next, given an entity from $\mathcal{T}$, we extract all of its properties from $\mathcal{G}$.  
We then organize all the extracted properties into a table, setting a null value to all properties whose values were missing. 
This process is equivalent to building the \emph{universal relation}~\cite{fagin1982simplied} out of all of the entity specific relations that were derived from $\mathcal{G}$. 

To extract more attributes and potentially improve the explanations, one may "follow" links in $\mathcal{G}$. Namely, extract also properties of values which are entities in $\mathcal{G}$ as well. This process can be done up to any number of hops in $\mathcal{G}$. All properties are then flattened and stored as a single table.

\textbf{Accommodating One-to-Many Relations}: 
\revised{The process described above assumes
that each entity is associated with a single value. However, real-world data often contain multiple categorical values (see Example \ref{ex:data_extraction}). Because correlation
is only defined for sets of paired values, downstream applications typically aggregate the values into a single number~\cite{santos2021correlation}. \sysName\ supports any user-defined function (e.g., mean, sum, max, first, or any representation-learning-based technique~\cite{bengio2013representation}) to perform the aggregation. }

\begin{example}
\label{ex:data_extraction}
A country's leader is an attribute extracted for each country. We can extract properties of the leaders, such as their age and gender, adding to $\mathcal{E}$ additional properties such as \textsc{Leader Age}, and \textsc{Leader Gender}. Other properties may point to multiple entities. The \textsc{US} entity has the property \textsc{Ethnic-Group}, which points to different ethnic groups. Each ethnic group is also an entity, and has the property \textsc{Population size}. One may add the property \textsc{Avg Population size of Ethnic-Group} to $\mathcal{E}$ by averaging the population sizes.   
\end{example}

\subsection{Handling Missing Data}
\label{subsec:IPW}
\revised{Extracted attributes, especially ones from KGs where data is sparse, may contain missing values. Our goal is to develop a principled approach to ensure the generated explanations are robust to missing data. }
Handling missing data is an enduring problem for many systems \cite{efron1994missing}.
The simplest approach to dealing with missing values is to restrict the analysis to complete cases, i.e., discard cases that have missing values.
However, this can induce \emph{selection bias} if the excluded tuples are systematically different from those included. For example, if the HDI values of only countries with a very high HDI are missing, restricting the analysis only to complete cases may lead to misleading explanations.  
A common solution is to impute missing values.  
Data imputation is unlikely to cause substantial bias if
few data are missing, but bias may increase as the number of missing data increases~\cite{seaman2013review}. 
Another common approach is Multiple Imputations (MI) \cite{patrician2002multiple}.
While MI is useful in supervised learning as long as it leads to models with an acceptable level of accuracy, 
MI makes a
missing-at-random assumption~\cite{efron1994missing}, which is often not the case in our setting.  
The approach that we followed is
Inverse Probability Weighting (IPW), 
a commonly used method to correct selection bias~\cite{seaman2013review}.  
In IPW, we consider only complete cases, but more weight is given to some complete cases than others. 
We next explain how to adapt IPW into our setting.

For simplicity of presentation, we assume that $\mathcal{T}$ and $\mathcal{E}$ have been joined into a single table. 
As we will explain in Section \ref{sec:MCIMR}, for an attribute $E {\in} \mathcal{E}$ we estimate $I(O;T|E,C)$ and $I(E;E')$ for $E' {\in} \mathcal{E}$. Therefore, we need to recover the probabilities $P(O|C,E), P(O|C,T,E), P(E)$, and $P(E|E')$. But since $E$ may contain missing values, we must ensure that those probabilities are \emph{recoverable}.  
Given an attribute $E$, let $R_E$ denote a selection attribute that indicates if the values of $E$ for the $i$-th tuple in the results of $Q$ is missing. I.e., $R_E[i] {=} 1$ if the value of $E$ for the $i$-th tuple was extracted, and $R_E[i] {=} 0$ otherwise. A complete cases analysis means that we examine only cases in which $R_E[i] {=} 1$.
Let $R_E {=} 1$ denote the selection of all tuples in which for them $R_E[i] {=} 1$ holds.  
We say the probability of an event $X$ which involves $E$ (e.g., $P(O|E)$) is recoverable if:
 $P(X) {=} P(X|R_E {=} 1)$.

We prove that $I(O;T|C,E)$
is recoverable if the complete cases are a representative sample of the original data, and each complete case is a random sample from the population of individuals
with the same $E$ and $T$ values. 
 \begin{proposition}
If $(O\independent R_E=1|E,C)$ and $(O\independent R_E=1|E,T,C)$, then $I(O;T|C,E)$ is recoverable. 
 \end{proposition}


We prove $I(E;E')$ 
is recoverable if the completeness of a case is independent of $E$, and remains independent given $E'$. 
 
\begin{proposition}
If $(E_i{\independent} R_{E_i}{=}1, R_{E_j} {=} 1)$ and $(E_i{\independent} R_{E_i}{=}1, R_{E_j} {=} 1|E_j)$, then $I(E;E')$ is recoverable.
 \end{proposition}

 

In situations other than described above, the probabilities will generally not be
recoverable. Following the IPW approach, we assign weights to complete cases, where the weight $W(X)$ of an event $X$ is 
defined as
$
 W(X) {=}$ $P(R_E{=}1)/P(R_E{=}1|X)   
$.
However, since $E$ contains missing values, $P(X)$ is unknown.  
We thus estimate $P(X)$. Commonly, a logistic regression model is fitted \cite{kang2007demystifying,hinkley1985transformation}.
Data available for this are the
values of the attributes in $\mathcal{D}$. We therefore employ a logistic regression (at pre-processing) to estimate $P(X)$. We note that although, as in MI, we predict missing values, 
we only use those predicted values for weights computation and not for the entire analysis. 
\vspace{-1mm}
\section{Algorithms}
\label{sec:MCIMR}

\vspace{-1mm}
\subsection{The MCIMR Algorithm}
\label{subsec:mcimr}
We present the MCIMR 
algorithm for the \probName\ problem. 
We show that MCIMR is 
a PTIME algorithm that finds the optimal $k$-size solution where $k$ is given. 
We then define a stopping criterion, allowing it to stop when no further improvement is found. 

\vspace{1mm}
When $k$ equals $1$, the optimal solution to \probName\ is the attribute $E {\in} \mathcal{A}$
that minimizes $I(O;T|C,E)$. When $k {\geq} 1$, a simple
incremental solution is to add one attribute at a time:
Given the explanation obtained at the $(k{-}1)$-th iteration $\boldsymbol{E_{k-1}}$, the $k$-th attribute to be added, denoted as $E_k$, is the one that contributes to the largest
decrease of $I(O;T|C,\boldsymbol{E_{k-1}})$. Formally,
\begin{align}\label{eq:obj}
E_k = argmin_{E {\in} \mathcal{A} \setminus \boldsymbol{E_{k-1}}} I(O;T|C,\boldsymbol{E_k})
\end{align}
where $\boldsymbol{E_k} {=} \boldsymbol{E_{k-1}} {\cup} \{E_k\}$.

It is difficult to get an accurate estimation for multivariate mutual information~\cite{peng2005feature}, as in Equation (1).
Instead, MCIMR calculates only bivariate probabilities, which is much more accurate, by incrementally selecting attributes based on Minimal-Conditional-mutual-Information (MCI) and Minimal-Redundancy (MR) criteria. 

The idea behind MCI is to search a $k$-size attribute set $\boldsymbol{E_k}$ that satisfies Equation \ref{eq:min_cmi}, which approximates Equation \ref{eq:obj}
with the mean value of all CMI values
between the individual attributes in $\boldsymbol{E_k}$ and $O$ and $T$:
\begin{align}\label{eq:min_cmi}
\boldsymbol{E_k} = argmin_{\boldsymbol{E_k} \subseteq \mathcal{A}} CI(O,T,C,\boldsymbol{E_k})
\end{align}
where $CI(O,T,C,\boldsymbol{E_k}) {=} \frac{1}{k} \sum_{E{\in} \boldsymbol{E_k}} I(O;T|C,E)$.

However, it is likely that attributes selected according to MCI are redundant. 
Thus, the following minimal redundancy condition is added: 
\begin{align}\label{eq:min_red}
\boldsymbol{E_k} = argmin_{\boldsymbol{E_k} {\subseteq} \mathcal{A}} Rd(\boldsymbol{E_k})
\end{align}
where $Rd(\boldsymbol{E_k}) = \frac{1}{k^2} \sum_{E_i, E_j{\in} \boldsymbol{E_k}} I(E_i;E_j)$.

Our goal is to minimize CI and Rd
simultaneously. Namely, we look for a $k$-size attribute set $\boldsymbol{E_k^*} {\subseteq} \mathcal{A}$ such that:
\begin{align}\label{eq:obj_mrmr}
\boldsymbol{E_k^*} = argmin_{\boldsymbol{E_k} \subseteq \mathcal{A}}\text{  } [CI(O,T,C,\boldsymbol{E_k}) + Rd(\boldsymbol{E_k})]
\end{align}

The MCIMR algorithm selects attributes incrementally as follows (as defined in Equation \ref{eq:obj_mrmr}). In the $k$-th iteration we
have the $k{-}1$-size attribute set $\boldsymbol{E_{k{-}1}}$.
The $k$-th attribute to be added is the attribute that minimizes the following condition:
\begin{align}
\label{eq:condition}
E_k {=} argmin_{E {\in} \mathcal{A}{\setminus} \boldsymbol{E_{k{-}1}}}[I(O;T|C,E) {+} \frac{1}{k{-}1}\sum_{E_i {\in} \boldsymbol{E_{k{-}1}}}\!\!\!\!I(E;E_i)]
\end{align}
We prove that the combination of the MCI and MRd criteria is equivalent to Equation \ref{eq:obj}. \revised{Namely, the MCIMR algorithm correctly computes the optimal $k$-size solution.}
\begin{theorem}
\label{thr:min_cmi_min_red}
The MCIMR algorithm yields the optimal $k$-size solution to Equation~\ref{eq:obj}.
\end{theorem}

\setlength{\textfloatsep}{2px}
\SetInd{1.0ex}{1.0ex}
\begin{algorithm}[t]
\scriptsize
	\DontPrintSemicolon
	\SetKwInOut{Input}{input}\SetKwInOut{Output}{output}
	\LinesNumbered
	\Input{A number $k$, a set of attributes $\mathcal{A}$, the outcome, treatment attributes $O$ and $T$, and the context $C$}
	\Output{An explanation $\boldsymbol{E}$.} \BlankLine
	\SetKwFunction{NextBestAtt}{\textsc{NextBestAtt}}
	
\textsc{MCIMR}($k, \mathcal{A}, O,T,C$):\\	
  $\boldsymbol{E} \gets \emptyset$.\\
   \For{ $i \in [1,k]$ }
   {
   $E_i \gets $ \NextBestAtt($O,T,C,\boldsymbol{E}, \mathcal{A}$)\\
   \If{$O \independent E_i | \boldsymbol{E}$ \tcp*{\textcolor{blue}{The responsibility test for $E_i$}}} 
   {
    \Return $\boldsymbol{E}$\\
   }
   $\boldsymbol{E} \gets \boldsymbol{E} \bigcup \{E_i\}$\\
   }

 \Return $\boldsymbol{E}$\\
 
\NextBestAtt($O,T,C,\boldsymbol{E}, \mathcal{A}$):\\
$E^* \gets $None, $v \gets \infty$\\
\ForEach{$E \in \mathcal{A}\setminus \boldsymbol{E}$}{
  \tcc{\textcolor{blue}{Weights are added if selection bias was detected}}
$v_1 \gets I(O;T|C,E), v_2 \gets 0$ \tcp*{\textcolor{blue}{Min CI computation}}
\ForEach{$E' \in \boldsymbol{E}$}
{
  \tcc{\textcolor{blue}{Weights are added if selection bias was detected}}
$v_2 \gets v_2 + I(E;E')$ \tcp*{\textcolor{blue}{Min redundancy computation}}
}
\If{$v_1 + \frac{v_2}{|\boldsymbol{E}|} < v$}{
$E^* \gets E, v \gets v_1 + \frac{v_2}{|\boldsymbol{E}|}$\\
}
}
\Return $E^*$\\
	\caption{The MCIMR Algorithm.}\label{algo:mcimr}
\end{algorithm}

\vspace{1mm}
\textbf{Stopping Criteria}. 
\label{subsec:findingk}
Up until this point we assumed that the size of the explanation $k$ is given. 
However, given two consecutive solutions of sizes $k$ and $k{+}1$, we can not say if $I(O;T|C,\boldsymbol{E_k}) {<}$ $I(O;T|C,\boldsymbol{E_{k+1}})$ or vice versa. 
As mentioned, we assume that
attributes in which their marginal explanation power is small are of no interest to analysts. 
We thus stop the algorithm after the first iteration in which the responsibility of the new attribute to be added is ${\approx} 0$. 
Namely, we treat $k$ as an upper bound on the explanation size.
To this end, we propose the \emph{responsibility test}. Given the set of selected attributes $\boldsymbol{E_k}$, this test verifies if the responsibility of a candidate attribute $E_{k+1}$ is ${\approx} 0$. 
\begin{lemma}[Responsibility test]
If $O {\independent} E_{k+1}|\boldsymbol{E_k}$ then $Resp(E_{k+1}) {\leq} 0$.
\end{lemma}
We measure conditional independence using the highly efficient independence
test proposed in~\cite{salimi2018bias}.




\vspace{1mm}
The MCIMR algorithm is depicted in Algorithm \ref{algo:mcimr}.
First, it initializes the attribute set $\boldsymbol{E}$ to be returned with the empty set (line~2). Then, new attributes are iteratively added according to the \textsc{NextBestAtt} procedure (line~4). The algorithm then applies the responsibility test to a selected attribute. If the responsibility of this attribute is ${\approx} 0$, the algorithm terminates and returns the solution obtained until this point (lines~5-7).
Otherwise, it terminates after $k$ iterations (line~9). Given the attribute set selected up until the $i$-th iteration, the \textsc{NextBestAtt} procedure finds the $i$-th attribute to be added. It implements Equation \ref{eq:condition}, by iterating over all candidate attributes and computing their individual explanation power (line~14), and their redundancy with selected attributes (lines~16-18). 
For simplicity, we omitted parts dedicated to handling missing data from presentation. In our implementation, before executing lines $14$ and $18$, we check if weights are needed to be added and adjust the computation accordingly. 

\begin{proposition}
\revised{The time complexity of the incremental MCIMR algorithm is $O(k |\mathcal{A}|)$}.
\end{proposition}
\revised{The size of $\mathcal{A}$ is potentially very large. Thus, in the next section, we propose several optimizations to reduce it.}

\subsection{Pruning Optimizations}
\label{subsec:optimizations}
We propose several optimizations to reduce the size of $\mathcal{A}$ and thereby reduce execution times. These optimizations are used to prune attributes that are either uninteresting as an explanation or cannot be a part of the optimal solution, and significantly improve running times.
We propose two types of optimizations: \textbf{Across-queries optimizations} that could be executed at pre-processing; and \textbf{Query-specific optimizations} that could be done only once $O$ and $T$ are known and are executed before running the MCIMR algorithm.

\textit{Preprocessing pruning}. 
Attributes discarded at this phase either have a fixed value, a unique value for each tuple, or lots of missing values. 
Thus, such attributes are uninteresting as an explanation \cite{salimi2018bias, li2021putting}. 
\textbf{Simple Filtering}:
We drop all attributes with a constant value (e.g., the attribute \textsc{Type} which has the value \textsl{Country} to all countries), and attributes in which the percentage of missing values is ${>}90\%$. 
\textbf{High Entropy}: we discard attributes such as \textsc{wikiID}, that have high entropy and (almost) a unique value for each tuple (as was done in \cite{salimi2018bias}). 

\textit{Online pruning}. 
\textbf{Logical Dependencies}: Logical dependencies can lead to a misleading conclusion that we found a confounding attribute, where we are, in fact, conditioning on an attribute that is functionally dependent on $T$ or $O$ (see proof in \cite{full}). 
We thus discard
all attributes $E$ s.t. $H(T|E) {\approx} H(E|T) {\approx}  0$ (resp., for $O$). These tests correspond to approximate functional dependencies~\cite{salimi2018bias}, such as \textsc{CountryCode} {$\Rightarrow$} \textsc{Country}.
\textbf{Low Relevance}:
As mentioned, we assume that the optimal explanation does not contain attributes which are individually unimportant but become important in the context of others. We leverage this assumption to prune attributes in which their individual explanation power is low (tested using conditional entropy, see full details in \cite{full}).


Another possible optimization is to cluster attributes that are highly correlated, such as \textsc{HDI} and \textsc{HDI Rank}, to reduce the redundancy among attributes~\cite{li2021putting}. However, we found this optimization to be not useful because of: 
(1) It could only be done after the query arrives, namely after we are done filtering, and the clustering process took longer than running our algorithm on all attributes. (2) We found that attributes clustered together were not necessarily semantically related. 



\subsection{Identifying Unexplained Subgroups}
\label{subsec:query_ref}
The MCIMR algorithm finds the explanations for the correlation between $T$ and $O$. While the generated explanation is optimal considering the whole data, it may be insufficient for some parts in the data. We thus propose an algorithm the analyst may use after getting the explanation, to identify unexplained data subgroups. It receives the original query $Q$ and its generated explanation. The output is a set of data groups correspond to 
context refinements of $Q$, in which a different explanation is required and thus may be of interest to the analyst.

\begin{example}
\label{ex:query_ref}
Consider a query compare the average salary of developers among countries. 
The explanation found by \sysName\ is $\boldsymbol{E} {=}  \{$\textsc{HDI, Gini}$\}$. As mentioned, the HDI of all countries in Europe is similar.
Thus, for countries in Europe, it is likely that $\boldsymbol{E}$ is not a satisfactory explanation.
\end{example}
For simplicity, numerical attributes are assumed to be binned.  
Data groups are defined by a set of attribute-value assignments and correspond to refinements of the context $C$ of $Q$. Treating the context $C$ as a set of conditions, a refinement $C'$ of $C$ is a set s.t. $C' {\subset} C$.
We search for the largest data groups s.t. $\boldsymbol{E}$ can not serve as their explanation. Formally, given an explanation $\boldsymbol{E}$, $I(O;T|C,\boldsymbol{E})$ is referred to as the explanation score for $C$. We are inserted in the top-$k$ data groups (in terms of size), each correspond to a context refinement $C'$ of $C$, s.t. their explanation score is ${>} \tau$ for some threshold $\tau$ ($\tau$ can be set based on the initial explanation score).

\begin{example}
\label{ex:query_ref2}
Continuing with Example \ref{ex:query_ref}, we refine $Q$ by adding a WHERE clause selecting only countries in Europe ($C' =\{$\textsc{Continent = Europe}$\}$). 
Let $Q_{EU}$ denote this query. We get: $I(O;T|C',\boldsymbol{E}) {=} 2.13$. As mentioned in Example $\ref{ex:running_ex3}$, the optimal explanation for $Q_{EU}$ is $\{$\textsc{Gini, Density}$\}$. 
\end{example}

A naive algorithm would traverse over all
possible contexts refinements $C'$, check if the explanation score is ${>} \tau$, and will choose
the largest data groups for which $\boldsymbol{E}$ is not a satisfactory explanation. We propose an efficient algorithm, exploiting the notion of pattern graph traversal~\cite{AsudehJJ19}. Intuitively, the set of all  context refinements can be represented as a graph 
where nodes correspond to refinements and there is an edge between $C$ and $C'$ if $C'$ can be obtained from $C$ by adding a single value assignment.
This graph can be traversed in a top-down fashion, while generating each node at most once (see \cite{full}). 

Algorithm~\ref{algo:topk} depicts the search for the largest $k$ data groups that for which $\boldsymbol{E}$ is not a satisfactory explanation. It traverses the refinements graph in a top-down manner, starting for the children of $C$. It uses a max heap $MaxHeap$ to iterate over the refinements by their size. It first initialize the result set $\mathcal{R}$ (line~\ref{line:initR}) and $MaxHeap$ with the children of $C$ (line~\ref{line:initHeap}). Then, while the $\mathcal{R}$ consists of less then $k$ refinements (line~\ref{line:whileStart}), the algorithm extracts the largest (by data size) refinement $C'$ (line~\ref{line:getMax}) and computes $I(O;T|C',\boldsymbol{E})$. If it exceeds the threshold $\tau$ (line~\ref{line:if}), $C'$ is used to update $\mathcal{R}$ (line~\ref{line:updateR}). The procedure \texttt{update} checks whether any ancestor of $C'$ is already in $\mathcal{R}$ (this could happen because the way the algorithm traverses the graph). If not, $C'$ is added to $\mathcal{R}$. 
 If $I(O;T|C',\boldsymbol{E}) {\leq} \tau$ (line~\ref{line:if}), the children of $C'$ are added to the heap (lines~\ref{line:insertFor}--~\ref{line:insertToHeap}).

\begin{proposition}
\label{prop:algo_topk}
\revised{
Algorithm~\ref{algo:topk} yields the top-$k$ largest data groups in which their explanation score is grater than $\tau$. }
\end{proposition}

\revised{In the worst case, there are no such $k$ data groups and hence the algorithm traverses over every possible context refinement of $Q$, which is polynomial in the number of attributes and (binned) values. However, as we show, in practice this algorithm efficiently identifies the data groups of interest, while exploring only an handful of context refinements. }

\setlength{\textfloatsep}{2px}
\SetInd{1.3ex}{1.3ex}
\begin{algorithm}[t]
\scriptsize
	\DontPrintSemicolon
	\SetKwInOut{Input}{input}\SetKwInOut{Output}{output}
	\LinesNumbered
	\Input{A number $k$, a set of attributes $\mathcal{A}$, the attributes $O$ and $T$, the context $C$, an explanation $\boldsymbol{E}$, and a threshold $\tau$.}
	\Output{Context refinements $\{C_1,{\ldots},C_k\}$ s.t. the corresponding groups are the largest $k$ groups and $I(O;T|C_i,\boldsymbol{E}) {>} \tau$} \BlankLine
    \SetKwFunction{GenChildren}{\textsc{GenChildren}}
    \SetKwFunction{update}{\textsc{update}}
    $\mathcal{R} \gets \emptyset$ \label{line:initR}\\
    $MaxHeap \gets \GenChildren(C)$ \label{line:initHeap}\\
    \While{$|\mathcal{R}| < k$ or $MaxHeap.isEmpty()$}
    {\label{line:whileStart}
        $C' \gets MaxHeap.exatractMax()$\label{line:getMax}\\
        \eIf{$I(O;T|C',\boldsymbol{E}) > \tau$}%
        {\label{line:if}
            $\update(\mathcal{R}, C')$\label{line:updateR} \tcp*{\textcolor{blue}{If none of the ancestors of $C'$ are in $\mathcal{R}$, insert $C'$ into $\mathcal{R}$.}}
        }
        {
            \For{$C'' \in \GenChildren(C')$ }
             {\label{line:insertFor}
                $MaxHeap.insert(C'')$\label{line:insertToHeap}
             }
        }
        
    }
    \Return $\mathcal{R}$\label{line:returnTopk}
	\caption{Top-$k$ unexplained data groups.}\label{algo:topk}
\end{algorithm}

\section{Experimental study}
\label{sec:experiments}
We present experiments that evaluate the effectiveness and
efficiency of our solution. We aim to address the following research questions. $Q1$: What is the
quality of our explanations, and how does it compare to that of existing methods? $Q2$: How robust are the explanations to missing data? $Q3$ What is the efficiency of the proposed algorithm and the optimization techniques? $Q4$: How useful are our proposed extensions? 


\vspace{1mm}
Our code and datasets are available at \cite{full}. We used DBPedia KG \cite{dbpedia} for attribute extraction, 
and the Pyitlib library~\cite{pyitlib}
for information-theoretic computations.  
The experiments were executed on a PC with a
$4.8$GHz CPU, and $16$GB memory.

\begin{table}
	\centering
	\scriptsize
		\caption{Examined Datasets.}
			\label{tab:data}
			\vspace{-8px}
	\begin{tabular} {|p{17mm}|p{8mm}|p{5mm}|p{42mm}|}
		\hline
	\textbf{Dataset} & \textbf{n}& \textbf{|$\mathcal{E}$|}&\textbf{Columns used for extraction} 
	 \\
		\hline
SO~\cite{stackoverflow}&47623&461&Country, Continent\\
				\hline
				COVID-19~\cite{covid19} &188&463&Country, WHO-Region\\
				\hline
				Flights~\cite{flights}&5819079&704&Airline, Origin/Destination city/state\\
				\hline
				Forbes~\cite{forbes} & 1647&708&Name\\
				\hline
	\end{tabular}
\end{table}

\paragraph*{Datasets}
We examine four commonly used datasets: 
\textbf{(1) SO}: Stack Overfow's annual developer survey is a survey of people who code around the world.
It has more than $47$K records containing information
about the developers' such as their age, income, and country. 
\textbf{(2) Covid-19}: This dataset includes information such as 
number of confirmed, death, and new cases in 2020 across the globe. 
\textbf{(3) Flights Delay}: This dataset contains transportation statistics of over $5.8$M domestic flights operated by large air carriers in the USA. 
\textbf{(4) Forbes}: This dataset contains annual earning information of $1.6$K celebrities between $2005-2015$
It contains the celebrities' annual pay, and category (e.g., Actors, Producers).

The attributes used for property extraction and the number of extracted attributes in each dataset are given in Table~\ref{tab:data}. 

\paragraph*{Baseline Algorithms}
We compare \sysName\ against the following baselines:
\textbf{(1) Brute-Force}: The optimal solution according to Def. \ref{problem:problem1}. This algorithm implements an exhaustive search over all subsets of attributes. To make it feasible, we run it after employing our pruning optimizations.  
\textbf{(2) Top-K}: This naive algorithm ranks the attributes according to their individual explanation power (equivalent to Max-Relevance only). 
\textbf{(3) Linear Regression} (LR): This baseline employs the OLS method to estimate the coefficients of a linear regression describing the
relationship between the outcome and the candidate attributes. The explanations are defined as the top-$k$ attributes with the highest coefficients (s.t. the $p$ value is ${<}.05$). Note that Pearson's $r$ is the standardized slope of LR and thus can be viewed as part of our competing baselines.
 \textbf{(4) HypDB} \cite{salimi2018bias}: This system employs an algorithm for confounding variable detection based on causal analysis. 
The explanations are defined as the top-k attributes with the highest responsibility scores. 
\textbf{(5) \sysName$^-$}: Last, to examine how pruning affects the explanation, we examine the explanation generated by \sysName\ without the pruning optimizations.

We also examined the explanations generated by CajaDE \cite{li2021putting}, a system that generates query results explanations based on augmented provenance information.  
However, since in all cases, CajaDE generated explanations obtained the lowest scores, we omit its results from presentation. The reason for that CajaDE explanations are a set of patterns that are unevenly distributed among groups in the query results, which are independent of the outcome variable. Thus, it cannot generate explanations that explain the correlation between $T$ and $O$.



Unless mentioned otherwise, we set the maximal explanation size, $k$, to $5$ and extracted attributes for 1-hop in the KG. For a fair comparison, we run all baselines (except for \sysName$^-$) after employing our pruning optimizations.

\begin{table*}
	\centering
	\scriptsize
		\caption{User study: The \textcolor{red}{best} and \textcolor{blue}{second best} explanations are marked in red and blue, resp.}
			\label{tab:case_study}
			\vspace{-6px}
	\begin{tabular} {|p{9mm}|p{2mm}|p{23mm}||p{16mm}|p{22mm}|p{18mm}||p{18mm}|p{18mm}|p{18mm}|}
		\hline
\multicolumn{2}{|c|}{\textbf{Dataset}} & \textbf{Query}&\textbf{Brute-Force}& \textbf{\sysName-}&\textbf{\sysName}&\textbf{Top-K}&\textbf{LR}&\textbf{HypDB}\\
		\hline
\multirow{6}{*}{SO}&$Q_1$ & Average salary per country &-&\textcolor{red}{\textbf{ HDI Rank, Gini}}&\textcolor{blue}{\textbf{HDI, Gini}} &HDI, Established Date& Population Census, Language &GDP\\\cline{2-9}

&$Q_2$ & Average salary per continent &-&\textcolor{blue}{\textbf{GDP Rank, Density}}&\textcolor{red}{\textbf{GDP,Density}} &GDP,Area rank&GDP, Area Rank&GDP\\\cline{2-9}

&$Q_3$ & Average salary per country in Europe&-&\textcolor{blue}{\textbf{Population Census, Gini Rank}} &\textcolor{red}{\textbf{Population Census, Gini}}&Population Census, Population Estimate&Population Census, Language &Gini, Area Rank\\

\hline\hline

\multirow{15}{*}{Flights}&$Q_1$ & Average delay per origin city & -&\textcolor{red}{\textbf{Precipitation Days, Year UV, Airline}} &Population urban, Year Low F, Airline&Year Low F, Year Avg F, December Low F&Year Low F, December percent sun, Day&\textcolor{blue}{\textbf{Year Low F, May Precipitation Inch, Airline}}\\\cline{2-9}

&$Q_2$ & Average delay per origin state &-& \textcolor{blue}{\textbf{Density, Year Snow, Airline}}&
\textcolor{red}{\textbf{Population estimation, Year Low F, Airline}}
&Population estimation, Population Urban, Population Rank&
Population estimation, Median Household Income, Distance
&Record Low F, Population estimation, Day\\\cline{2-9}

&$Q_3$ &Average delay per origin cities in CA &-& \textcolor{blue}{\textbf{ Density, Population Metropolitan, Security Delay}} & Density, Population Total,Security Delay&Population Metropolitan,Security Delay &-&\textcolor{red}{\textbf{Density, Population Ranking, Cancelled}}\\\cline{2-9}

&$Q_4$ &Average delay per origin state and airline&- &\textcolor{blue}{\textbf{Population Total, Fleet size}}&\textcolor{red}{\textbf{Population Ranking, Fleet size}} &Density, Population Total &- &Revenue, Dec Record Low F\\\cline{2-9}

&$Q_5$ &Average delay per airline &-& \textcolor{red}{\textbf{Equity, Fleet Size}}&\textcolor{red}{\textbf{Equity, Fleet Size}} &Equity, Net Income&\textcolor{red}{\textbf{Equity, Fleet Size}}&\textcolor{blue}{\textbf{Num of Employees, Revenue}}\\
\hline\hline

\multirow{9}{*}{Covid-19}&$Q_1$ & Deaths per country&\textcolor{red}{\textbf{HDI, GDP, Confirmed cases}} &\textcolor{blue}{\textbf{HDI, GDP Rank, Confirmed cases}}&\textcolor{red}{\textbf{HDI, GDP, Confirmed cases}}&GDP Rank, GDP Nominal, HDI&Area Rank, Currency, Recovered cases &Density, Time Zone, Confirmed cases\\\cline{2-9}

&$Q_2$ &Deaths per country in Europe & \textcolor{blue}{\textbf{Gini, Population Census, Confirmed cases}}&\textcolor{red}{\textbf{ Gini Rank, Density, Confirmed cases}} &\textcolor{blue}{\textbf{Gini, Population Census, Confirmed cases}}&Gini Rank, Gini, GDP&Area Rank, Currency, Population Total&Currency, GDP, New cases\\\cline{2-9}

&$Q_3$ & Average deaths per WHO-Region &\textcolor{red}{\textbf{Density, Confirmed Cases}}& \textcolor{red}{\textbf{Density,Confirmed Cases}}&\textcolor{red}{\textbf{Density,Confirmed Cases}}&\textcolor{red}{\textbf{Density,Confirmed Cases}}&-&\textcolor{blue}{\textbf{Area Km,Confirmed Cases}}\\
\hline
\hline

\multirow{6}{*}{Forbes}&$Q_1$ & Salary of Actors&\textcolor{red}{\textbf{Net Worth, Gender, Age}} &\textcolor{blue}{\textbf{Net Worth, ActiveSince, Gender}} &Net Worth, Gender&Net worth, Awards&Citizenship, Honors&Gender, Honors\\\cline{2-9}

&$Q_2$ & Salary of Directors/Producers &\textcolor{blue}{\textbf{Net Worth, Awards}}&\textcolor{red}{\textbf{Years Active, Net Worth}} & \textcolor{blue}{\textbf{Net Worth, Awards}}&Net Worth, Age&-&Years Active\\\cline{2-9}

&$Q_3$ & Salary of Athletes &\textcolor{red}{\textbf{Cups, Draft Pick, Active Years}}&National Cups, Draft Pick  &Cups, Draft Pick&Total Cups, National Cups&-&\textcolor{blue}{\textbf{Cups, Active Years}}\\
\hline

 	\end{tabular}
 	\vspace{-4mm}
\end{table*}

\begin{table}
	\centering
	\scriptsize
		\caption{Avg. explanation scores according to the subjects.} 
			\label{tab:scores_users}
			\vspace{-6px}
	\begin{tabular} {|p{15mm}|p{15mm}|p{20mm}|}
		\hline
	\textbf{Baseline} & \textbf{Average Score}& \textbf{Average Variance}
	 \\
		\hline
		Brute-Force&3.8&0.8\\
		\sysName-&3.7&1.1\\
		\sysName&3.5&0.9\\
		HypDB&2.8&1.1\\
		Top-K&2.1&0.8\\
		LR&1.8&0.6\\
		\hline

	\end{tabular}
\end{table}

\subsection{Quality Evaluation (Q1)}
\label{subsec:quality}
We validate our intuition that attributes extracted from KGs can explain correlations in common scenarios. To this end, we randomly generated $40$ SQL queries (10 from each dataset) as follows. We set $T$ to be one of the attributes used for attribute extraction (as listed in Table \ref{tab:data}).
We set $O$ to be a numerical attribute that could be predicted from the data (e.g., \textsc{Departure/Arrival Delay} in Flights, \textsc{New/Death Cases} in Covid-19). We then added a WHERE clause by randomly picking another attribute and one of its values, ensuring selected subsets contain more than 10\% of the tuples in the original dataset. Full details are given in the Appendix. We say our approach was useful if (1) the partial correlation between the exposure and outcome (while conditioning on an explanation generated by \sysName) is lower than the original correlation, and (2) the explanation contains at least one extracted attribute. We report this was the case in $72.5\%$ percent of the queries. 

Next, we aim to asses the quality of the generated explanations to validate our problem definition. To this end, 
we present a user study consisting of explanations produced by each algorithm. Since a standard benchmark for results explanation does not exist, we consider $14$ representative queries suffering from confounding bias, as shown in Table \ref{tab:case_study}.
Our queries are inspired by real-life sources, such SO annual reports~\cite{stackoverflowreport}, news and media websites (e.g., Vanity Fair~\cite{vanityfair}, USA Today\cite{usatoday} for Forbes and Flights), and academic papers \cite{upadhyay2021correlation,kaklauskas2022effects}. 
Similar experiments were conducted in \cite{salimi2018bias,li2021putting, lin2021detecting}. 
To compare the generated explanations with the "ground-truth" explanations, we will show that our explanations are supported by previous findings. A similar approach was taken in \cite{salimi2018bias}. 

We
recruited $150$ subjects on Amazon MTurk. This sample size enables us to observe a 95\% confidence level with a 10\% margin of error.  
Subjects were asked to rank each explanation of each method (shown together with its corresponding query) on a scale of $1{-}5$, where $1$ indicates that it does not make sense and $5$ indicates that the explanation is highly convincing. The form we gave to the subjects is available at \cite{full}. 

HypDB's time complexity is exponential in the size of
$\mathcal{A}$~\cite{salimi2018bias}.
We run it over all attributes in $\mathcal{A}$ (after pruning) and report that it never terminates within $10$ hours.
Thus, 
we have no choice but to limit the number of attributes for HypDB, to allow it to generate explanations in a reasonable time. For HypDB, besides pruning, we omitted candidate attributes uniformly at random, ensuring that $|\mathcal{A}| {\leq} 50$. 
We only report the results of Brute-Force for the small Covid-19 and Forbes datasets, as it was infeasible to compute them for the larger datasets. We do not randomly drop attributes for computational efficiency here because Brute-Force is intended to be an optimal solution for our problem definition against which our algorithm is judged. 
The explanations generated by different methods are given in Table~\ref{tab:case_study}, and the average explanation scores given by the subjects are depicted in Table~\ref{tab:scores_users}.

We summarize our main finding as follows:\\
$\bullet$ The subjects found the explanations generated by Brute-Force, \sysName$^-$, and \sysName\ to be the most convincing. This supports our mathematical definition (Def 2.1) of what constitutes a good explanation.\\
$\bullet$ \sysName\ explanations are supported by previous in-domain findings, which serve as "ground-truth" explanations. \\
$\bullet$ Our pruning has little effect on explanation quality. \\
$\bullet$ The next best competitor is HypDB. However, it is unable to scale to a large number of candidate attributes.\\
$\bullet$ As expected, Top-k yields redundancy in selected attributes.

First, subjects found the explanations generated by Brute-Force, \sysName$^-$, and \sysName\ to be the most convincing.  
The pairwise differences between the average scores of these 3 methods are not statistically significant. Previous in-domain findings also support these explanations. For example, in SO $Q_1$, it was shown in~\cite{stackoverflowreport} that there is a correlation between developers salary and countries' economies (reflected in the HDI and Gini values). For Flights $Q_1$, it was stated in~\cite{usatoday} that weather is one of the top reasons for flights delay in the US. For Covid-19 $Q_1$, it was shown that there is a correlation between countries' economies and Covid-19 death rate \cite{upadhyay2021correlation,kaklauskas2022effects}. More details can be found in the Appendix.
In all cases where the results of Brute-Force and \sysName\ are different, it happens because \sysName\ drops attributes with insignificant responsibility (according to the responsibility test). For example, in Forbes $Q_1$, \sysName\ dropped \textsc{Age}. 
The low difference between the results of \sysName$^-$ and \sysName\ indicates that pruning has little effect on explanations quality. Namely, \emph{\sysName\ is able to execute efficiently without compromising on explanation quality.}


The explanations of all methods consist of attributes extracted from the KG. This validates our assumptions that KGs can serve as valuable sources for results explanations. 
The next best competitor is HypDB (the average score is worse than that of \sysName. This difference is statistically significant, $p {<} .05$). This is not surprising as HypDB finds confounding attributes using causal analysis. However, its main disadvantage is its ability to scale for large number of attributes. In cases where HypDB generated explanations that were considered not convincing, it was mainly because important attributes were dropped (as we limited the number attributes to enable feasible execution times). 
Not surprisingly, the explanations generated by Top-K and LR were considered to be less convincing (their average scores are statistically significant from all other methods, $p{<}.05$). 
For Top-K, this is substantially because it ignores redundancy among attributes. For example, in Flights $Q_1$, it chose the attributes \textsc{Year Low F} and \textsc{Year Average F}, which are highly correlated. For LR, in many cases, it failed to generate explanations, as there were no attributes with low enough p-values. Even when it succeeded, the subjects found them to be not convincing. The reason is that LR focuses on finding linear correlations.

\begin{figure}[t]
\begin{center}
		\includegraphics[scale = 0.4]{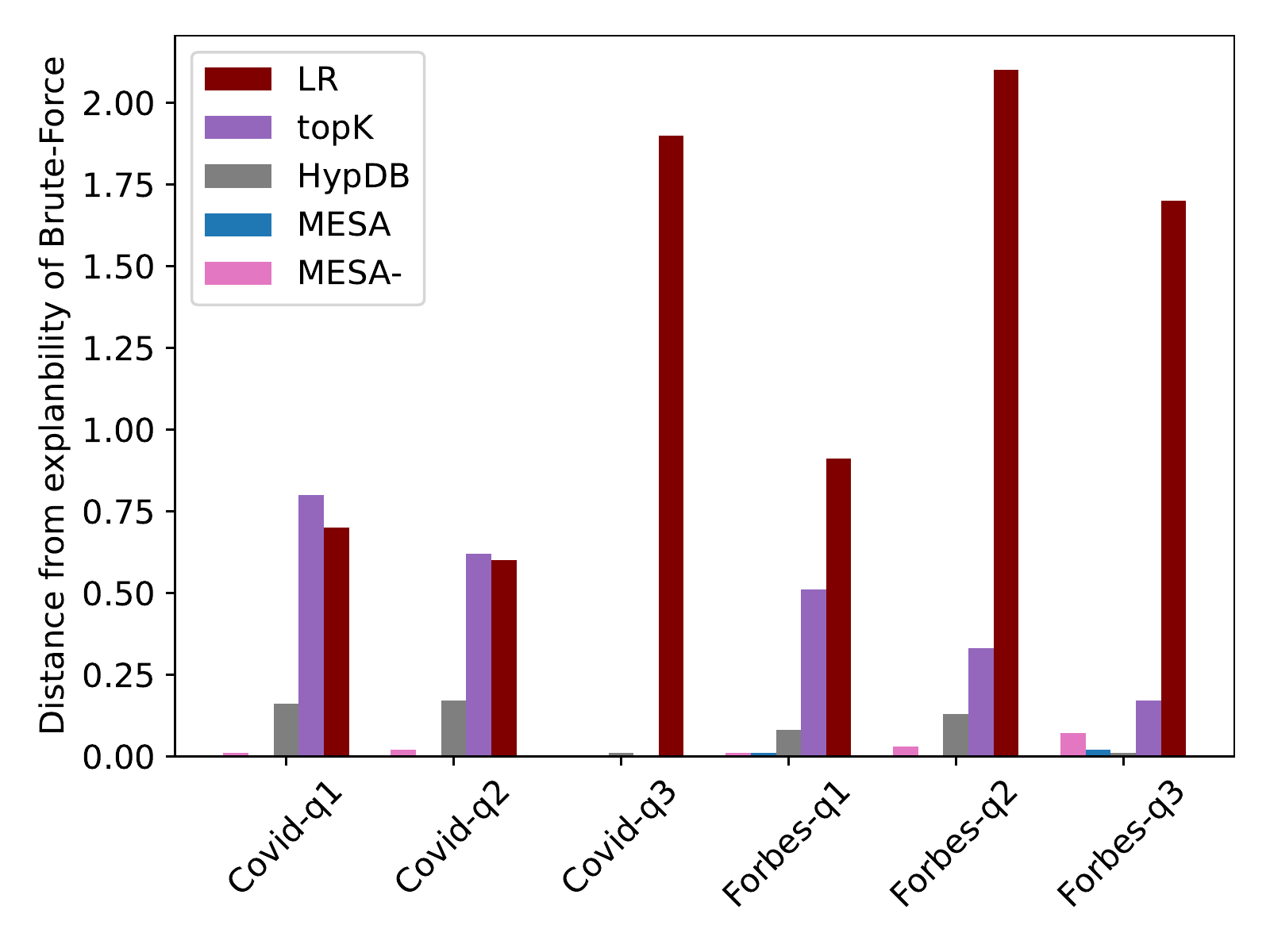}
		\vspace{-8px}
	
		\caption{Distance from explainability scores of Brute-Force.} 
			\label{fig:explainability_scores}
\end{center}
\end{figure}

\vspace{1mm}
\textit{Explainability scores}.  
Let $\boldsymbol{E}$ denote the explanation found by an algorithm. We call $I(O;T|\boldsymbol{E})$ the explainability score. Explainability score equal to $0$ means that $\boldsymbol{E}$ perfectly explains the correlation between $O$ and $T$. 
\revised{The explainability scores of Brute-Force serve as the gold standard (as by definition, it aims to minimize this score)}. 
In some cases, the explanations generated by all algorithms, including Brute-Force, cannot fully explain the correlations. E.g., in Flights $Q_2$, the explainability score of Brute-Force is $0.25$. This means that other factors that affect flight delays may not exist in the KG (e.g., labor problems).
The results are depicted in Figure \ref{fig:explainability_scores}.
The y-axis is the distance between the explainability scores of each method and Brute-Force. The lower the distance the better is the explanation.   
Observe that the explainability scores of \sysName\ are almost as good as the ones of Brute-Force and \sysName$^-$, and are much better than those of the competitors.

Additional experiments can be found in the Appendix.

\subsection{Robustness to Missing Data (\textbf{Q2})} 
On average, the percentage of missing values in extracted attributes is 37\%, 42\%, 45\% and 73\% in Covid-19, SO, Flights and Forbes, resp. The high prevalence of missing values in Forbes is because DBpedia uses different attributes to describe a person from each category (e.g., actors, authors). 
In Covid-19, SO, Flights, and Forbes, the percentage of attributes with selection bias is 13.3\%, 14.1\%, 24.2\%, and 29.4\%, resp. 
\emph{This verifies that selection bias exists in attributes extracted from KGs, and thus should be appropriately handled.}

We examine the robustness of our explanations to missing data, by varying the percentage of missing values from the top $10$ most relevant (w.r.t. the outcome) attributes. We examine two ways to omit values: missing-at-random and biased removal, where the top-$x$ highest values from examined attributes were omitted (when varying $x$). We examine the effect on our generated explanations average explainability score.
Explainability should not be affected if an explanation is robust to missing data. We also examine the effect on the explainability scores while imputing missing values (using the common mean imputation technique~\cite{zhang2016missing}). The results for the SO and Covid datasets are depicted in Figure~\ref{fig:missing_data}. As expected, data imputation has huge negative effect on explainability. Our approach is much less sensitive to missing data: Even with 50\% missing values (at random or not), the explainability scores have hardly changed. When the percentage of missing values is above 50\%, a lot of the information is lost, and thus it is harder to estimate partial correlation correctly. 

\begin{figure}[t]

	\begin{center}
		\begin{minipage}[b]{0.25\textwidth}
		\includegraphics[scale = 0.25]{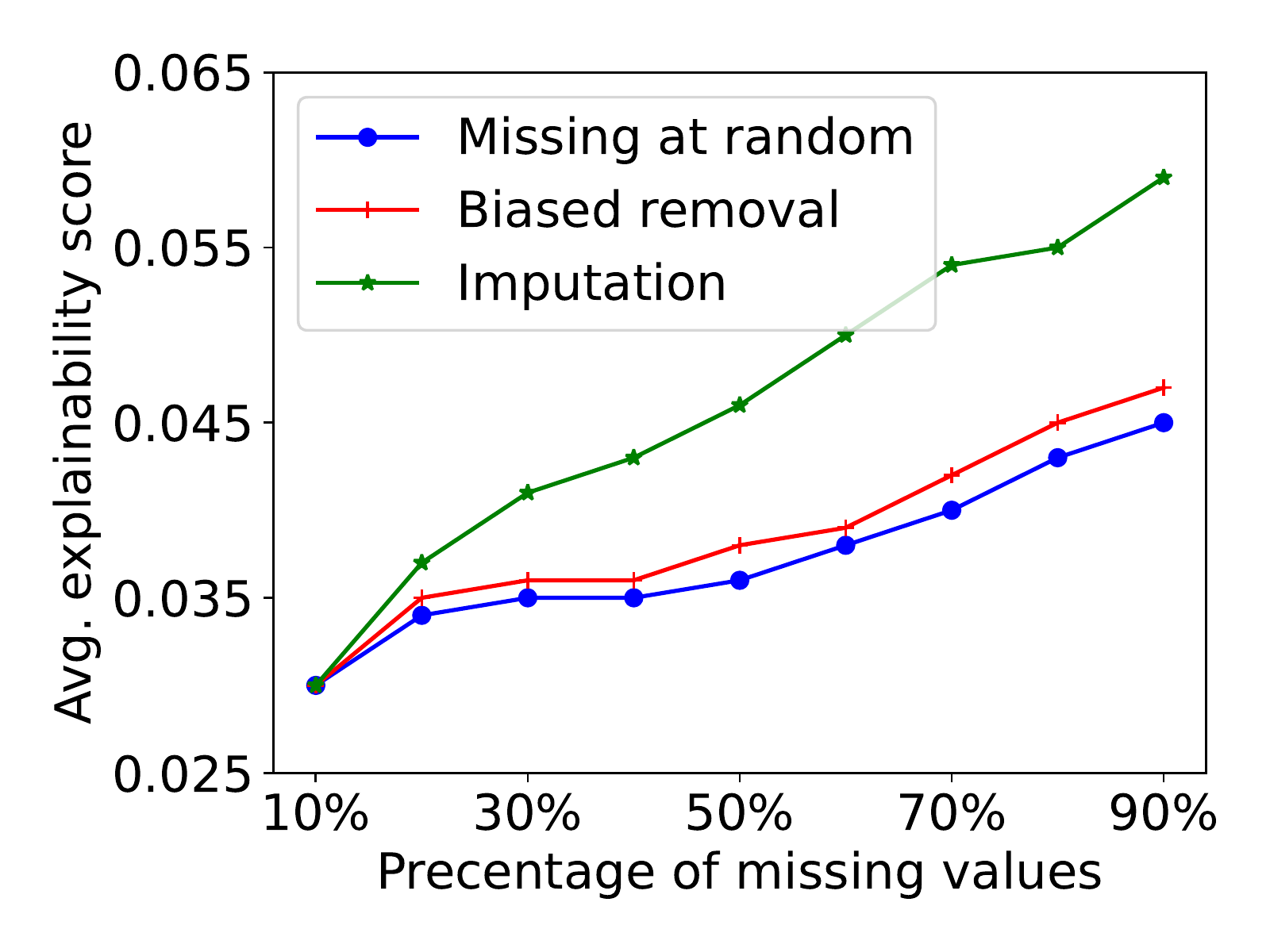}
		\vspace{-10px}
		\centering
		\caption*{{SO}}  
	\end{minipage}%
		\begin{minipage}[b]{0.25\textwidth}
		\includegraphics[scale = 0.25]{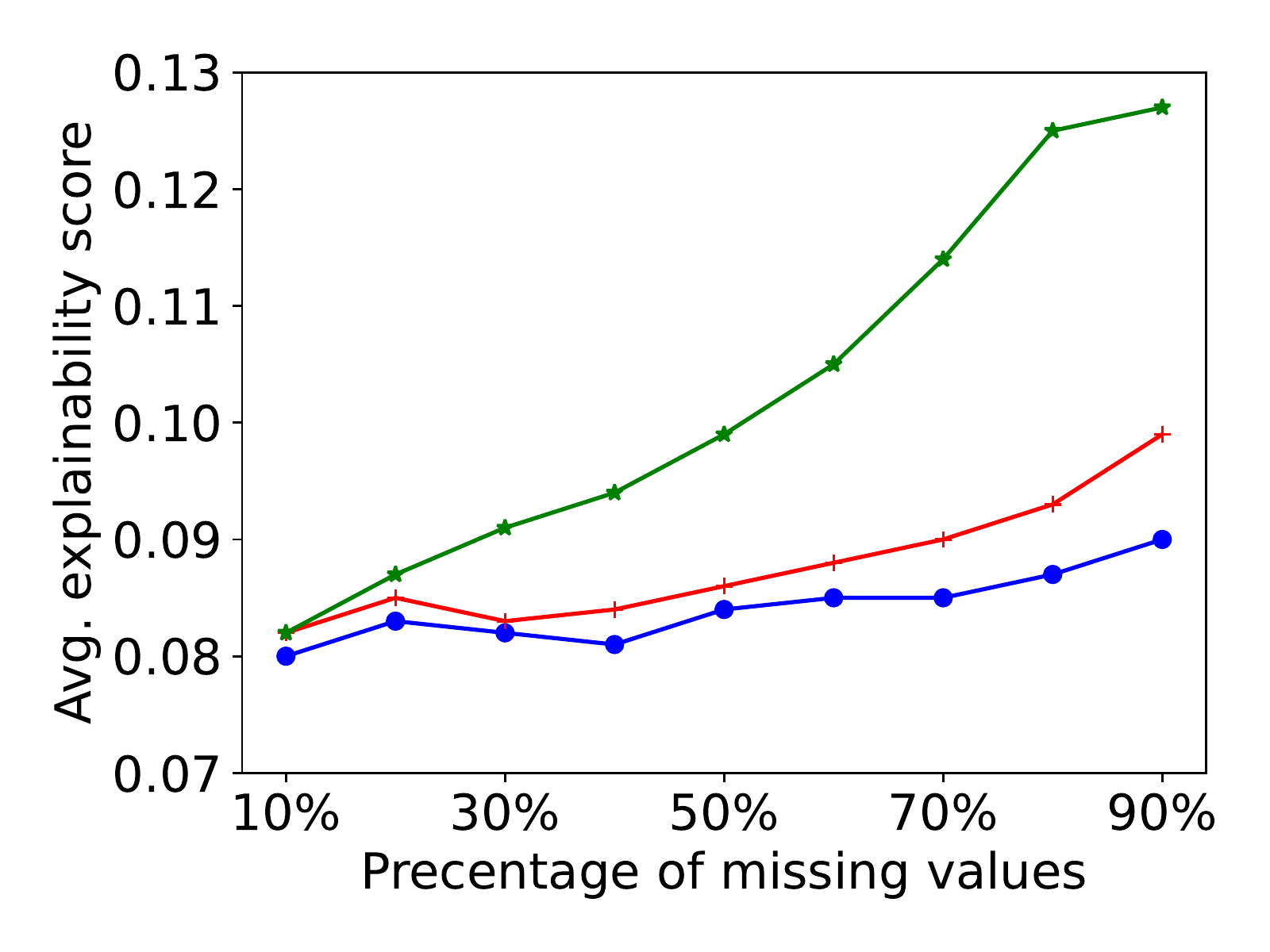}
		\vspace{-10px}
		\centering
		\caption*{{Covid-19}}  
	\end{minipage}%
		\end{center}
	\vspace{-10px}
	\caption{Explainability as a function of missing data.} \label{fig:missing_data}
	\vspace{-1mm}
\end{figure}


\begin{figure*}[htpb]
\centering
	\begin{center}
		\begin{minipage}[b]{0.333\textwidth}
		\includegraphics[scale = 0.3]{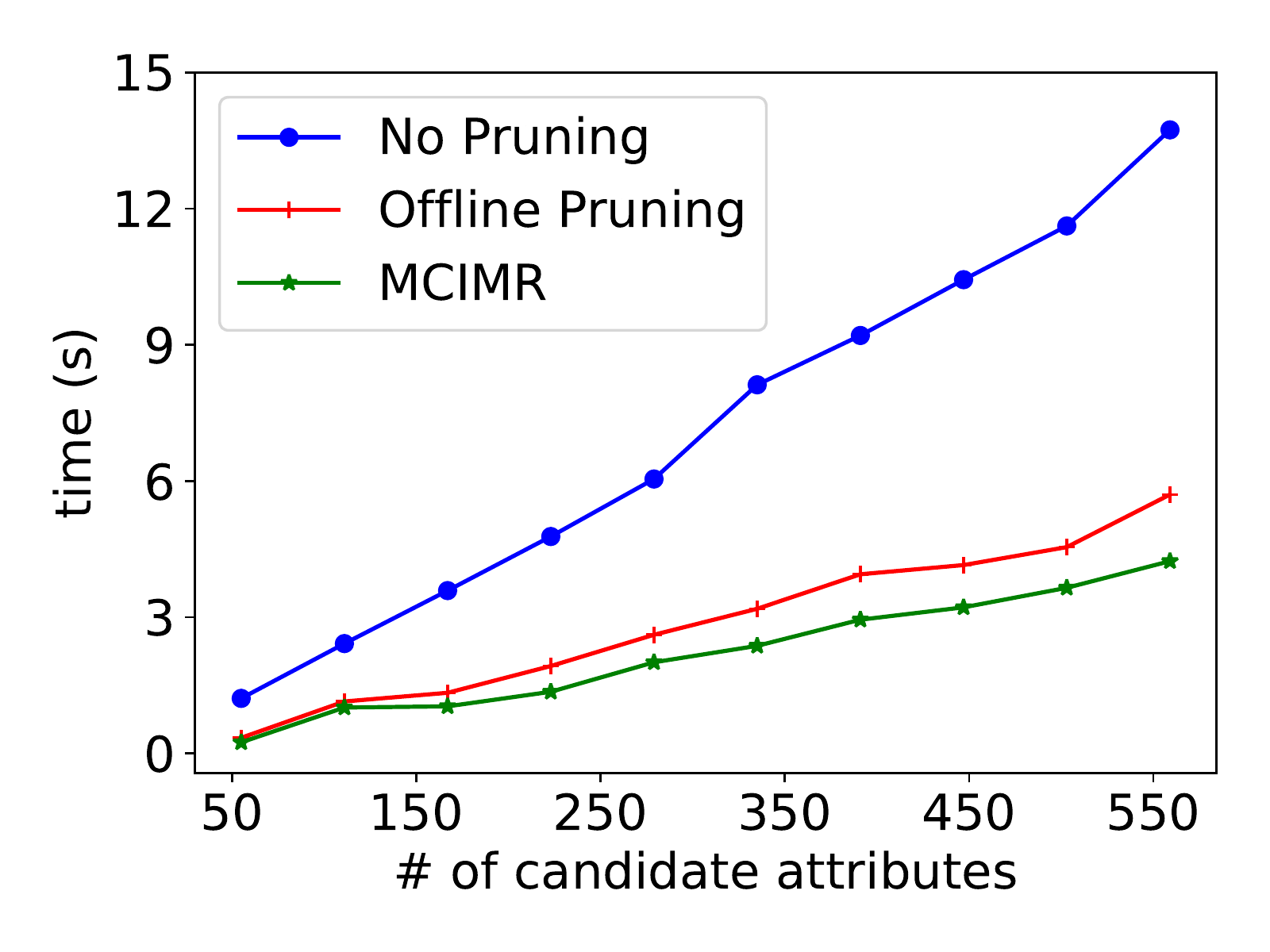}
		\vspace{-8px}
		\centering
		\caption*{{SO}}  
	\end{minipage}%
		\begin{minipage}[b]{0.333\textwidth}
		\includegraphics[scale = 0.3]{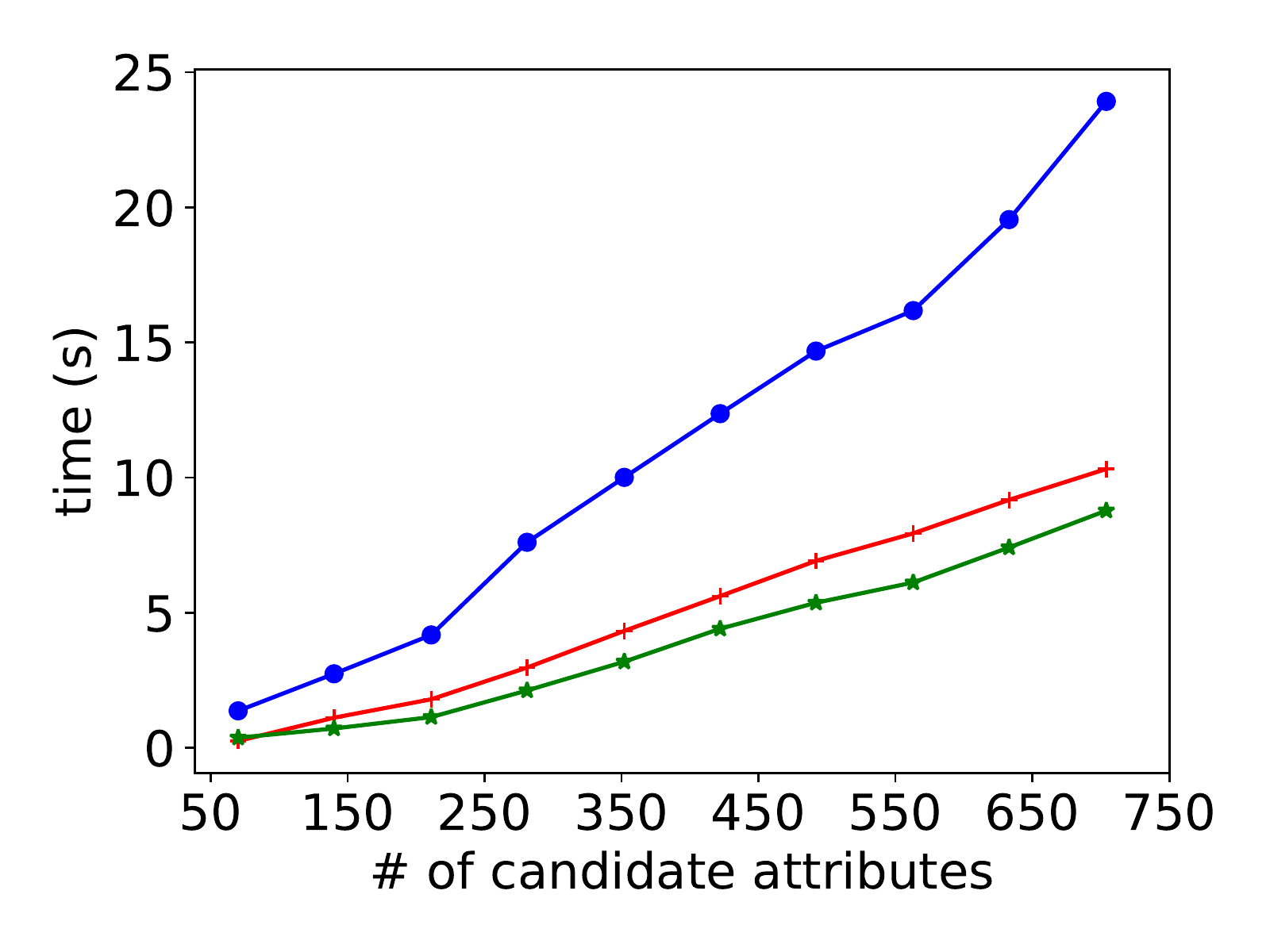}
		\vspace{-8px}
			\centering
		\caption*{{Flights}}  
	\end{minipage}%
		\begin{minipage}[b]{0.333\textwidth}
		\includegraphics[scale = 0.3]{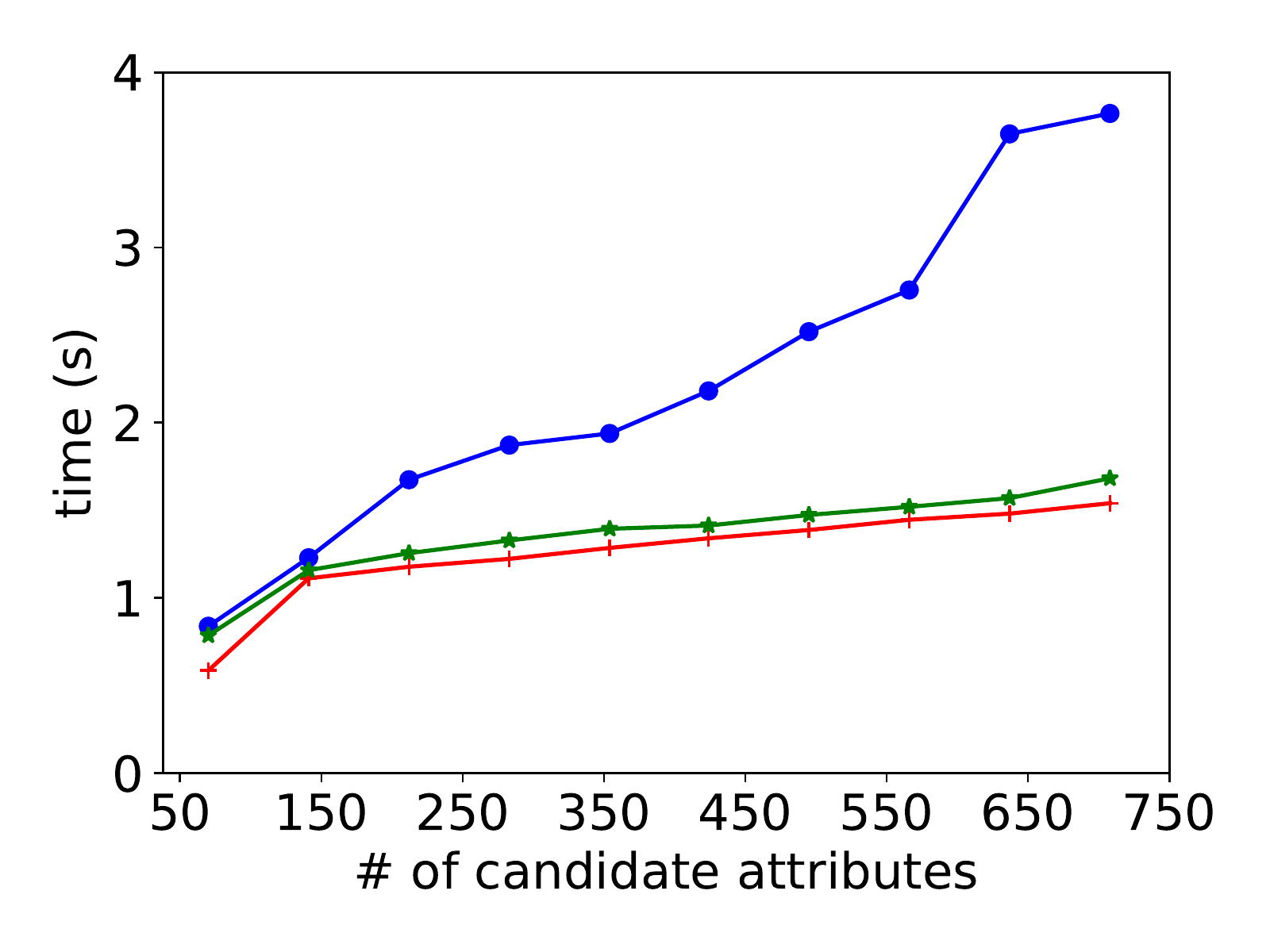}
		\vspace{-8px}
			\centering
		\caption*{{Forbes}}  
	\end{minipage}%
		\end{center}
	\vspace{-10px}
	\caption{Running times as a function of the number of candidate attributes.} \label{fig:times_cols}
	\vspace{-4mm}
\end{figure*}
\begin{figure*}[htpb]

	\begin{center}
		\begin{minipage}[b]{0.333\textwidth}
		\includegraphics[scale = 0.3]{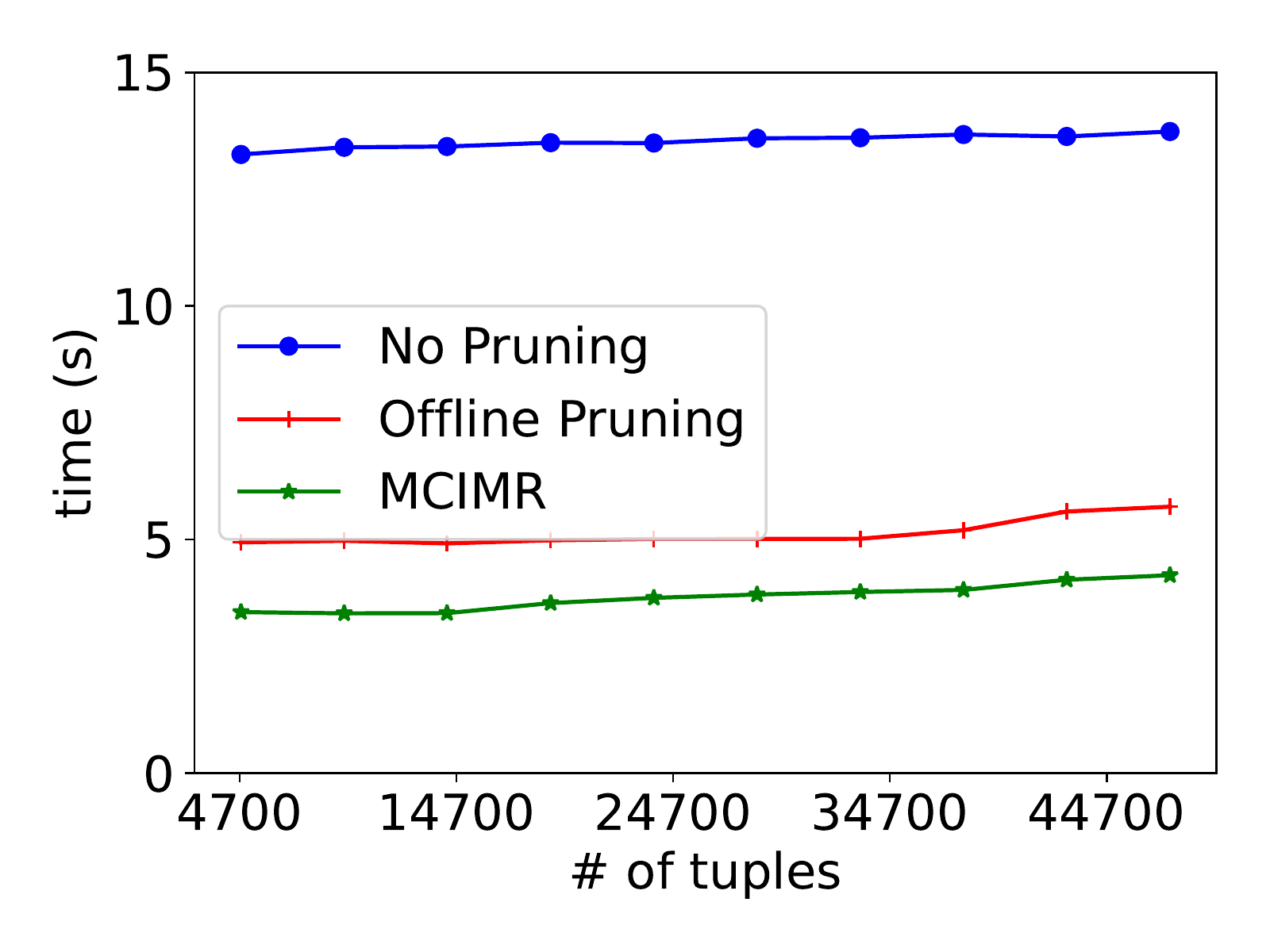}
		\vspace{-8px}
			\centering
		\caption*{{SO}}  
	\end{minipage}%
		\begin{minipage}[b]{0.333\textwidth}
		\includegraphics[scale = 0.3]{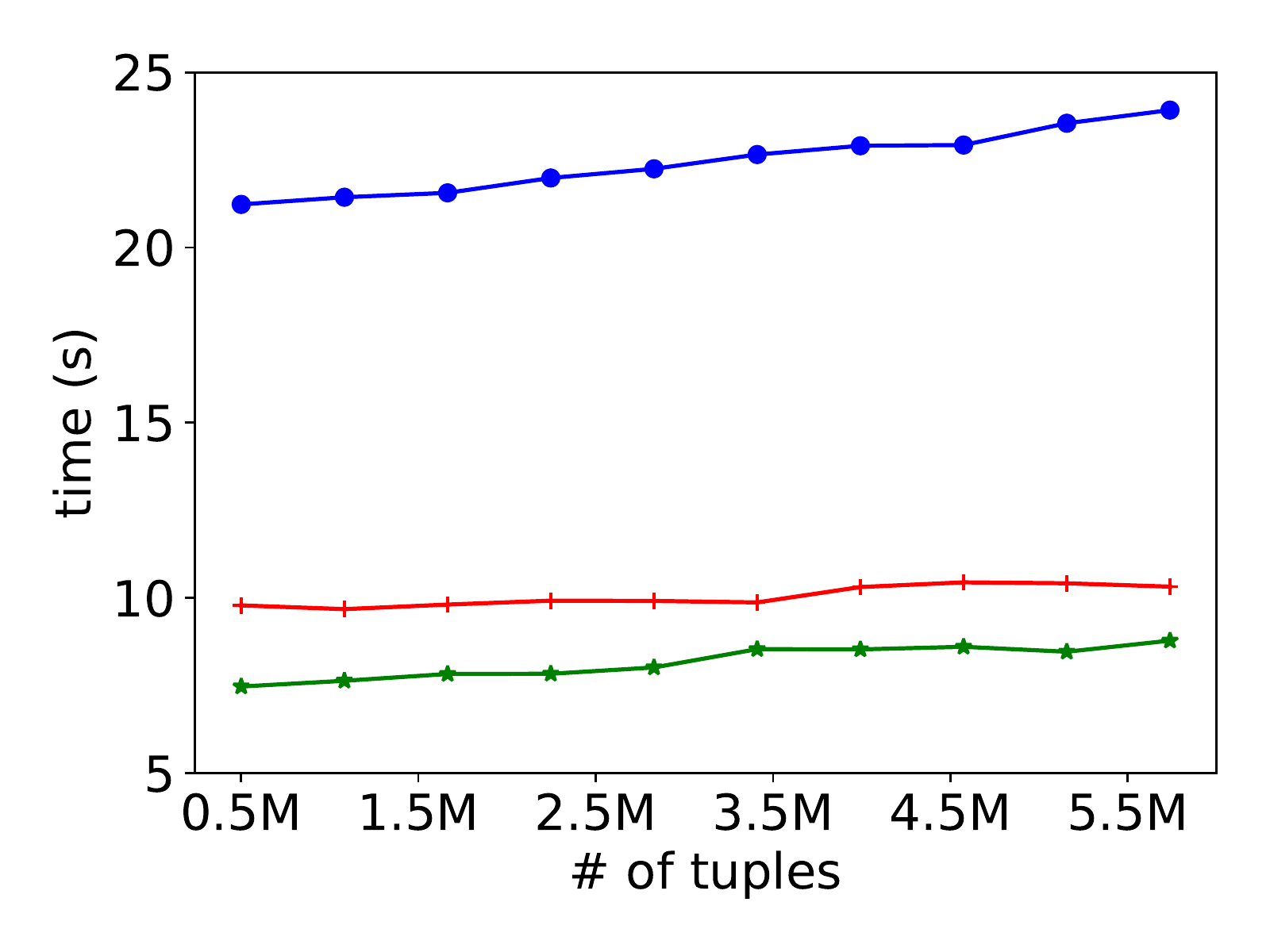}
		\vspace{-8px}
			\centering
		\caption*{{Flights}}  
	\end{minipage}%
		\begin{minipage}[b]{0.333\textwidth}
		\includegraphics[scale = 0.3]{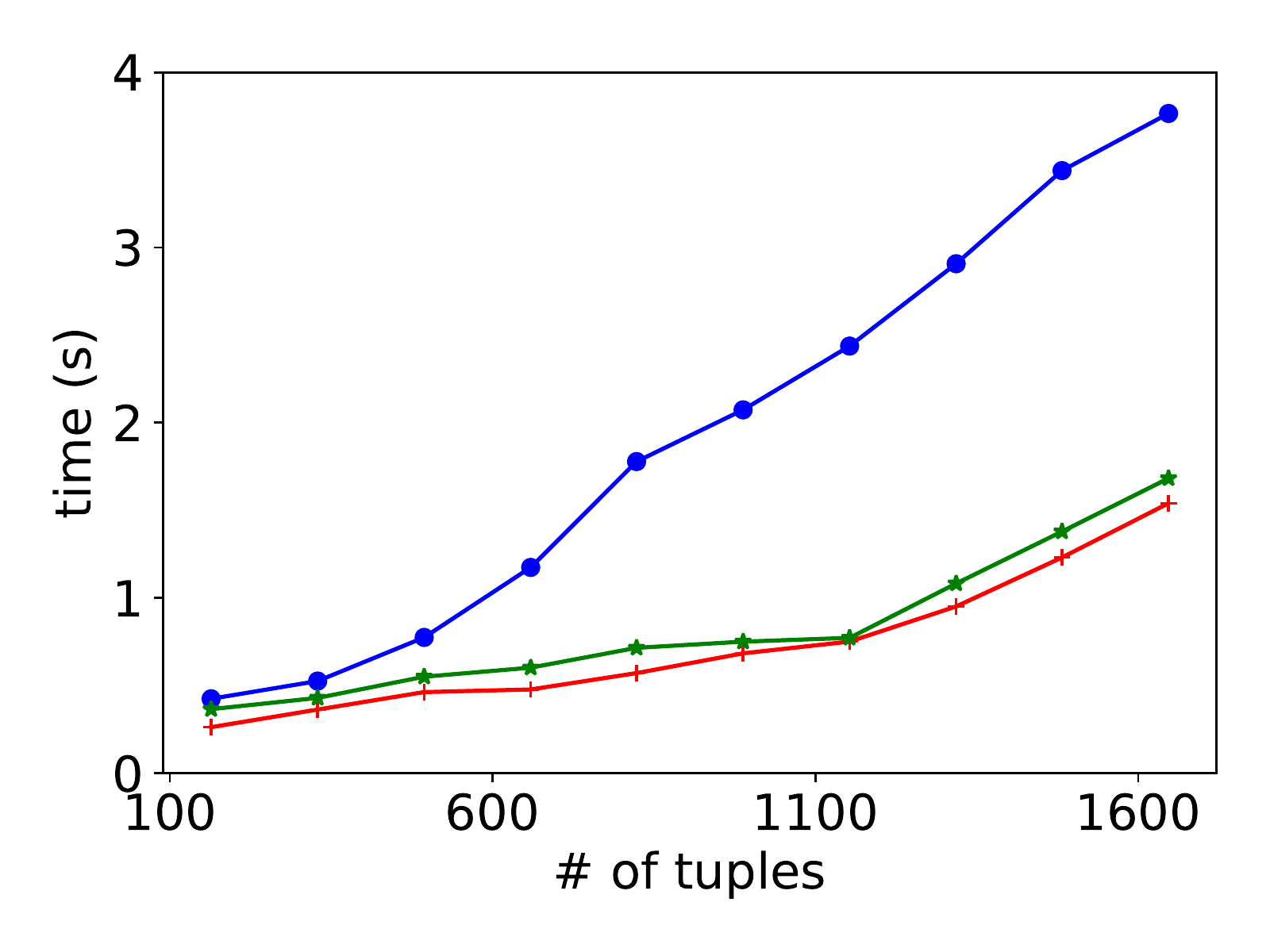}
		\vspace{-8px}
			\centering
		\caption*{{Forbes}}  
	\end{minipage}%
		\end{center}
	\vspace{-10px}
	\caption{Running times as a function of the number of rows in the dataset.} \label{fig:times_data}
	\vspace{-4mm}
\end{figure*}

\begin{figure*}[htpb]
	\begin{center}
		\begin{minipage}[b]{0.333\textwidth}
		\includegraphics[scale = 0.3]{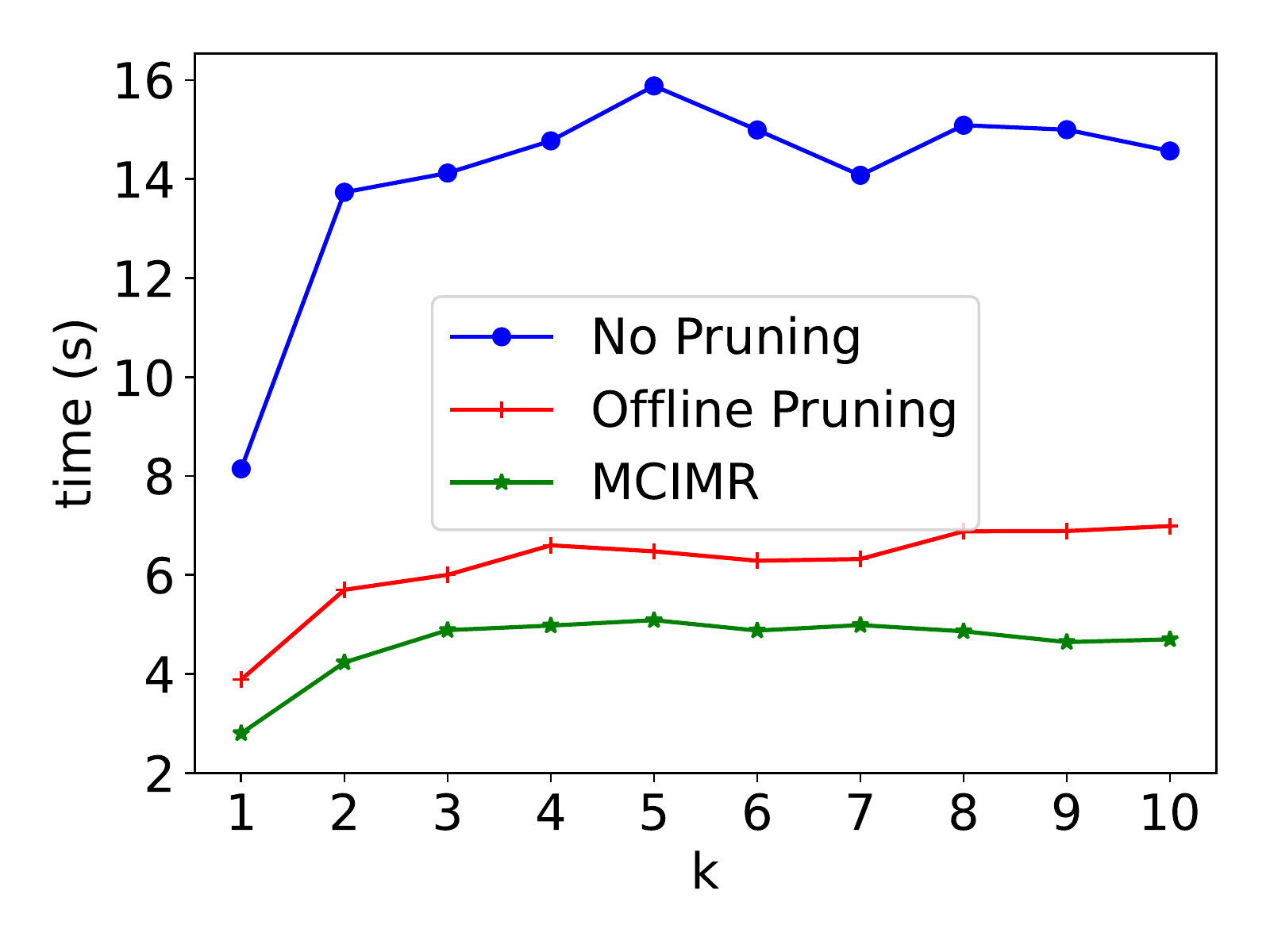}
		\vspace{-8px}
			\centering
		\caption*{{SO}}  
	\end{minipage}%
		\begin{minipage}[b]{0.333\textwidth}
		\includegraphics[scale = 0.3]{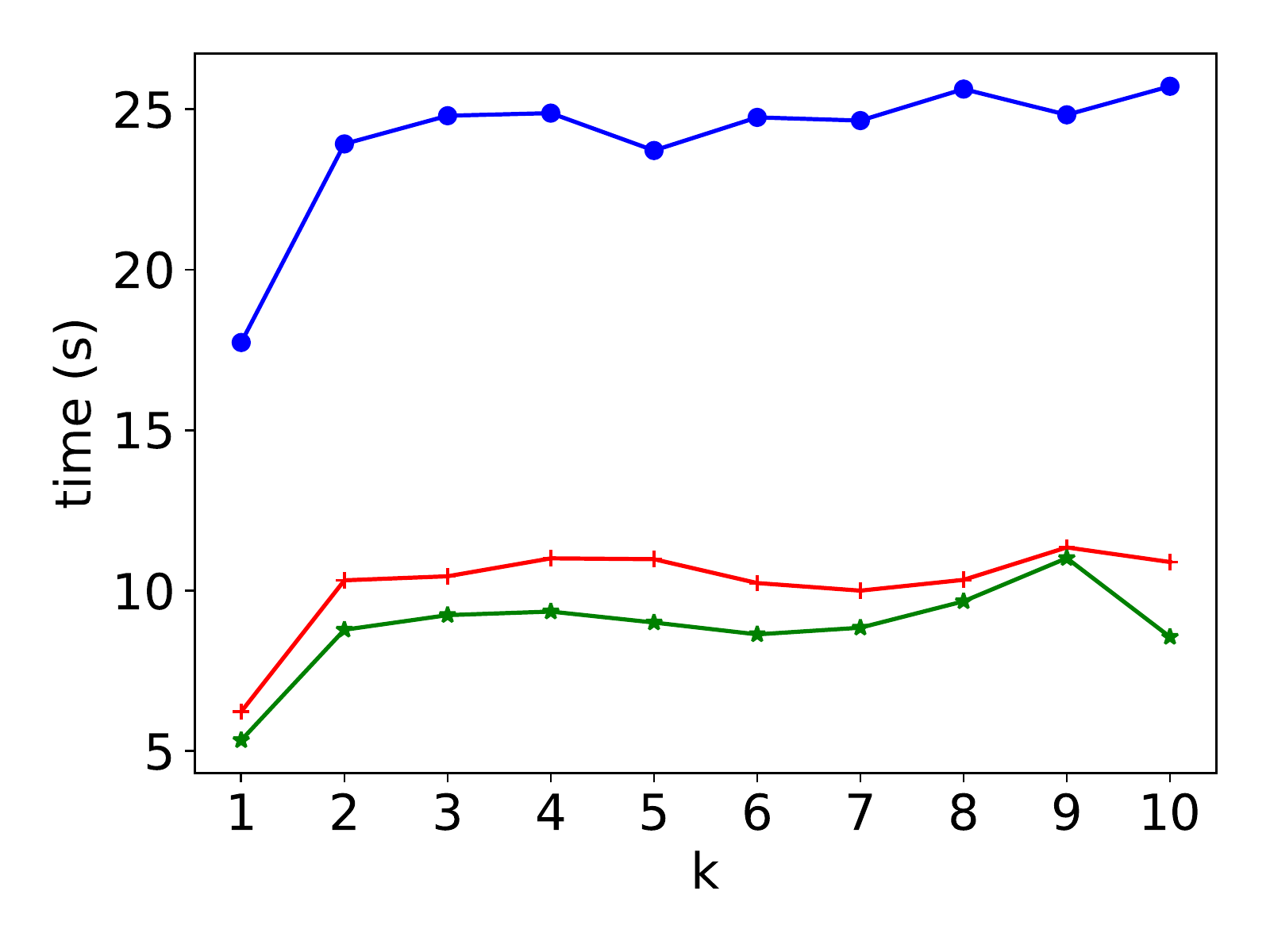}
		\vspace{-8px}
			\centering
		\caption*{{Flights}}  
	\end{minipage}%
		\begin{minipage}[b]{0.333\textwidth}
		\includegraphics[scale = 0.3]{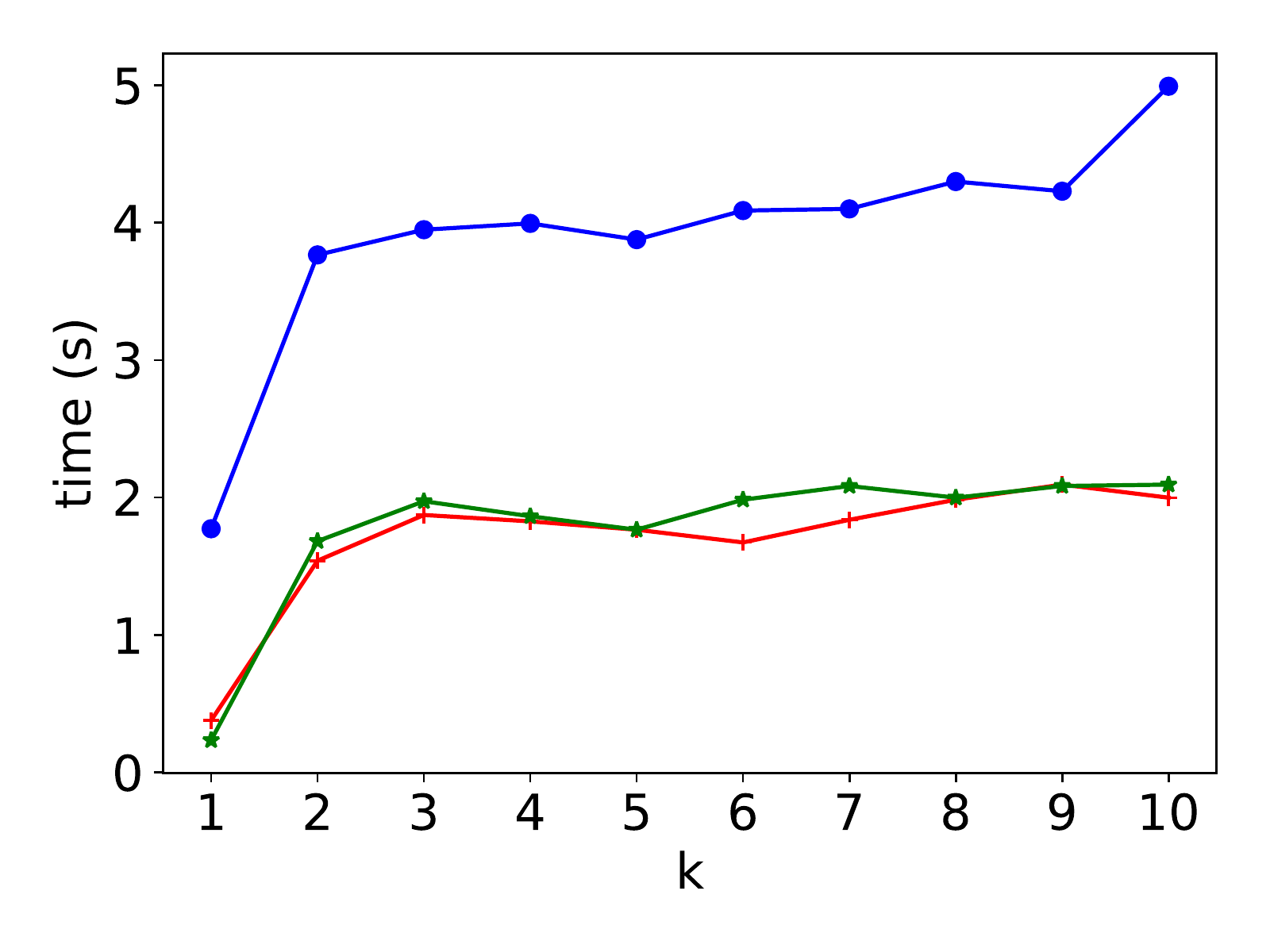}
		\vspace{-8px}
			\centering
		\caption*{{Forbes}}  
	\end{minipage}%
		\end{center}
	\vspace{-10px}
	\caption{Running times as a function of the bound on the explanation size.} \label{fig:times_k}
	\vspace{-5mm}
\end{figure*}

\subsection{Efficiency Evaluation (Q3)}
To examine the contribution of our optimizations, we report the running
times of the following baselines:
\textbf{No Pruning} \textemdash the MCIMR algorithm without pruning;
\textbf{Offline Pruning} \textemdash MCIMR with only offline pruning.
We study the effect of multiple parameters on
running times. For each dataset, we report the average execution time of the queries presented in Section \ref{subsec:quality}. 
In all cases, the execution time of MCIMR was less than $10$ seconds, a reasonable response time for an interactive system.
We omit the results obtained on the (smallest) Covid-19 dataset from presentation, as the results demonstrated similar trends to those of Forbes.

\vspace{1mm}
\textit{Candidate Attributes}.
In this experiment, we omitted from consideration attributes from $\mathcal{A}$ uniformly at random. The results are depicted in Figure \ref{fig:times_cols}. In all dataset, we exhibit a (near) linear growth in
running times as a function of the size of $\mathcal{A}$. 
The execution times of No-Pruning are significantly higher than those of Offline Pruning and MCIMR, \emph{indicating the usefulness of the offline pruning.} The difference in times across datasets is due to their size. Estimating CMI on large datasets (e.g., Flights, SO) takes longer than on small datasets (e.g., Forbes). 
In Forbes, Offline Pruning is faster than MCIMR, implying that in small datasets online pruning is not necessary, as it takes longer that running MCIMR.  


\vspace{1mm}
\textit{Data Size}. We vary the number of tuples in $\mathcal{D}$, by removing tuples uniformly at random. The results are depicted in Figure \ref{fig:times_data}. In SO and Flights, observe that the dataset size has a little effect on running times. This is because the size of the subgroups in the group-by queries were big. Thus, when randomly omitting tuples from the datasets, the number of considered groups is almost unchanged. On the other hand, since in Forbes each group contained only a few records, we exhibit a (near) linear growth in
running times.

\vspace{1mm}
\textit{Explanation size}. We vary the bound on the explanation size. Recall that given a bound $k$, MCIMR returns an explanation of size ${\leq}k$. It may return an explanation of size $l {<} k$ if the responsibility of the $l{+}1$ attribute is ${\approx}0$. The results are shown in Figure \ref{fig:times_k}. In all cases, the size of the explanations was no bigger than $3$. Thus, $k$ has almost no effect on running times, as the algorithms terminate after no more than $4$ iterations.

\subsection{Extensions (Q4)}
\label{subsec:q3}

We examine the effect of extracting attributes following more than one hop in the KG. We report that in the vast majority of cases, \sysName's explanations were unaffected, indicating that most of the relevant information can be found in the first hop. Further details can be found in the Appendix .


\vspace{1mm}
\textit{Unexplained Subgroups}.
We demonstrate the effectiveness of the Top-K unexplained groups algorithm by focusing on SO $Q_1$, setting $\tau {>} 0.2$. The top-5 largest unexplained data groups are given in Table~\ref{tab:unexplained}. Observe that economy-related attributes (e.g., GDP, HDI) of selected data groups are internally consistent (e.g., the HDI of countries in Europe is similar). Thus, it makes sense that the explanation for SO $Q_1$ ($\{$\textsc{HDI, Gini}$\}$) will not be a satisfactory explanation for these data groups. Indeed, as shown in Table \ref{tab:case_study}, the explanation of \sysName\ for the top-1 unexplained group (SO $Q_3$) is different from the one found for all countries. 
We ran this algorithm over the other queries as well. The average execution time is $4.4s$. This demonstrates the ability of our algorithm to efficiently identify data subgroups that are likely to be of interest to users.

\begin{table}
	\centering
	
	\scriptsize
		\caption{Top-5 unexplained groups for SO Q1.}
			\label{tab:unexplained}
			\vspace{-8px}
	\begin{tabular} {|p{7mm}|p{7mm}|p{48mm}|}
		\hline
\textbf{Rank} &	\textbf{Size} & \textbf{Data group}
	 \\
		\hline
		1&18342&\textsc{Continent} $=$ \textsc{Europe}\\
		2&17899&\textsc{Continent} $=$ \textsc{Asia}\\
		3&15466& \textsc{Continent} $=$ \textsc{North America}\\
		4&14788&\textsc{Currency} $=$ \textsc{Euro}\\
		5&12754&\textsc{Continent} $=$ \textsc{Africa}\\
		\hline

	\end{tabular}
\end{table}
\section{Related work}
\label{sec:related}

\textit{Results Explanations}.
Methods explaining why data is missing or mistakenly included in query results have been studied in~\cite{bidoit2014query,chapman2009not,lee2020approximate,ten2015high}. Explanations for unexpected query results have been presented in~\cite{bessa2020effective, miao2019going}. Those works are orthogonal to our work, as we aim to explain unexpected correlations. 
Another line of work provides explanations on how a query result was derived by analyzing its provenance and pointing out tuples that significantly affect the results~\cite{milo2020contribution,meliou2010complexity,meliou2009so}. Those methods are designed to generate tuple-level explanations and not attribute-level explanations that are required for unearthing correlations. 
Another type of explanation for query results is a set of patterns that are shared by one (group of) tuple but not by another (group of) tuple~\cite{el2014interpretable,li2021putting,roy2015explaining,roy2014formal,wu2013scorpion}. However, those works as well do not account for correlations among attributes.
\cite{salimi2018bias} presented HypDB, a system that 
identifies confounding bias in SQL queries for improved decision making, detected using causal analysis. However, as mentioned, HypDB only considers attributes from the input table, and it cannot efficiently handle a large amount of candidate attributes.

We share with~\cite{li2021putting} the motivation for considering explanations that are not solely
drawn from the input table. \cite{li2021putting} presented CajaDE, a system that generates explanations of query results based on
information from tables related to the table
accessed by the query. However, related tables often do not
exist. Moreover, as mentioned in Section \ref{sec:experiments}, their explanations are \emph{independent of the outcome}. Thus, even if CajaDE is given the attributes mined from other sources, it may generate explanations that are irrelevant to the correlation between the exposure and outcome.


\textit{Dataset Discovery}. \revised{Given an input dataset,
dataset discovery methods find related
tables that can be integrated via join or
union operations. Existing methods estimate how joinable or unionable two datasets are~\cite{yang2019gb,zhu2019josie,zhu2016lsh,nargesian2018table}. Other works focused on automating the data augmentation task to discover relevant features for ML models~\cite{chepurko2020arda}. While these works focus
on finding datasets that are joinable or unionable, we aim to find unobserved attributes that explain unexpected correlations. Recent work proposed solutions to discover  datasets that can be joined with an input dataset and contain a column that is correlated with a target column~\cite{esmailoghli2021cocoa,santos2021correlation}. Such techniques can be integrated into our system for
extracting candidate attributes from tabular data. We focus on finding attributes that minimize the partial correlation between two columns rather than finding columns that are correlated with a target column. Thus, future work will extend these techniques to support our goal. }



\textit{Feature Selection}. 
\revised{The \probName\ problem is closely related to the well-studied Feature Selection (FS) problem~\cite{chandrashekar2014survey,guyon2003introduction,li2017feature}.  FS methods select a concise and diverse set of attributes relevant to a target attribute for use in model construction~\cite{chandrashekar2014survey}, whereas we aim to select a conciseness and no-redundant set of attributes that are correlated to the outcome and exposure.  
Closest to our project is a line of work using information-theoretic methods for FS~\cite{li2017feature}. 
Algorithms in this family \cite{lin2006conditional,el2008powerful,meyer2008information,fleuret2004fast} define different criteria to maximize feature relevance and minimize redundancy. Relevance is typically measured by the feature correlation with the target attribute.
Of particular note, the MRMR algorithm~\cite{peng2005feature} selects features based on Max-Relevance and Min-Redundancy criteria. A main difference in MCIMR is that instead of the Max-Relevance criterion, we use a Min-CMI criterion. Another critical difference is the stopping condition. While in MRMR, the size $k$ of the selected feature set is determined using the underlying learning model, in MCIMR, we set $k$ using responsibility scores.}

\textit{Explainable AI}. 
A related line of work is Explainable AI (XAI), an emerging field in machine learning that aims to address how black box decisions of AI systems are made \cite{adadi2018peeking,dovsilovic2018explainable}.
Similar to our approach, XAI can be used to learn new facts,
to gather information and thus to gain knowledge \cite{adadi2018peeking}. We share the motivation with posthoc XAI methods~\cite{assche2007seeing,martens2007comprehensible,dong2017towards}, which extract explanations from already learned models. 
The advantage of this approach is that it does not impact the performance of the model, which is treated as a black box. In \sysName\ as well, we generate explanations after the SQL query was executed, independently from the database engine.

\section{Conclusion and Limitations}
\label{sec:conc}
\revised{This paper presented the \probName\ problem, whose goal is to identify uncontrolled confounding attributes that explain unexpected correlations observed in query
results. 
We developed an efficient algorithm that finds the optimal subset of attributes. 
This algorithm is embodied in a system called \sysName, which adapts the IPW technique for handling missing data.  
\sysName\ is applicable for cases where explanations can be found in a given external knowledge source. 
In this paper we focused on extracting attributes from KGs. Future work would extracted candidate attributes from other sources, such as unstructured data (e.g., text documents). Another interesting future work is to identify which links in a KG are relevant to the explanation and worthy to follow}.





\bibliographystyle{ACM-Reference-Format}
\bibliography{vldb_sample}

\appendix
\section{Missing Proofs}
\label{app:proofs}

In this part we provide missing proofs. 

\begin{proof}[Proof of Corollary \ref{cor:cor1}]
For any two random variables, if $X {\independent} Y$ we have: $H(Y|X) = H(Y)$. This can be generalized to conditional independence as well. We get:
\begin{multline*}
    I(O;T|E,R_E=1,C) = 
    H(O|E,R_E=1,C) - H(O|T,E,R_E=1,C) = \\H(O|E,C) - H(O|T,E,C) = I(O;T|E,C)
\end{multline*}
\end{proof}

\begin{proof}[Proof of Corollary \ref{cor:cor2}]
We have:
\begin{multline*}
    I(E_i; E_j|R_{E_i} = 1, R_{E_j} =1) = \\H(E_i|R_{E_i} = 1, R_{E_j} =1)- H(E_i|E_j,R_{E_i} = 1, R_{E_j} =1) =\\
    H(E_i) - H(E_i|E_j) = I(E_i; E_j)
\end{multline*}
\end{proof}

In what follows, to ease the exposition, we assume that there is no WHERE clause in the query, i.e., $C = \emptyset$. Our results also hold for cases where $C$ is not empty. 

\begin{proof}[Proof of Theorem \ref{thr:min_cmi_min_red}]
Recall that by definition of the algorithm, we assume
that $\boldsymbol{E_{k-1}}$, i.e., the set of $k{-}1$ attributes, has already been
obtained, and thus $\boldsymbol{E_{k-1}}, O$, and $T$ are fixed when selecting the $k$-th attribute. The goal is to select the optimal $k$-th attribute to be added, $E_k$,
from $\mathcal{D} \setminus \boldsymbol{E_{k-1}}$.

By the definition of conditional mutual information, we have:
\begin{multline*}
 I(O;T|\boldsymbol{E_{k-1}}, E_k) = I(O;T|\boldsymbol{E_k}) = \\H(O;\boldsymbol{E_k}) + H(T;\boldsymbol{E_k}) - H(O;T;\boldsymbol{E_k}) - H(\boldsymbol{E_k})
\end{multline*}

We use the following definition of \cite{peng2005feature} for the attributes $E_1, \ldots, E_k$: $J(\boldsymbol{E_k}) = J(E_1,\ldots,E_k)$ where:
$$J(\boldsymbol{E_k}) = \sum\ldots\sum Pr(E_1,\ldots,E_k) \frac{Pr(E_1,\ldots,E_k)}{Pr(E_1)\cdot \ldots\cdot Pr(E_k)}$$
Similarly, we have:
$$J(O,T, \boldsymbol{E_k}) {=} \sum{\ldots}\sum Pr(E_1,\ldots,E_k, O,T) \frac{Pr(E_1,\ldots E_k, O,T)}{Pr(E_1){\ldots} Pr(E_k) Pr(O) Pr(T)}$$
$$J(X, \boldsymbol{E_k}) = \sum{\ldots}\sum Pr(E_1,{\ldots}E_k, X) \frac{Pr(E_1,{\ldots}E_k, O)}{Pr(E_1)\cdot {\ldots} Pr(E_k) Pr(X)}$$

We can derive: 
\begin{multline*}
 H(O;\boldsymbol{E_k}) + H(T;\boldsymbol{E_k})   -H(O;T;\boldsymbol{E_k}) -  H(\boldsymbol{E_k}) =\\
 \\H(O) +\sum_{i = 1}^k  H(E_i)-J(O,\boldsymbol{E_k}) + H(T)+\sum_{i = 1}^k  H(E_i)- J(T,\boldsymbol{E_k}) \\ -H(O) -H(T) - \sum_{i = 1}^k  H(E_i) + J(O,T,\boldsymbol{E_k})  -  \sum_{i = 1}^k  H(E_i) + J(\boldsymbol{E_k}) = \\J(O,T,\boldsymbol{E_k}) + J(\boldsymbol{E_k}) - J(O, \boldsymbol{E_k}) - J(T, \boldsymbol{E_k})
\end{multline*}

Thus we consider the following expression:
\begin{align}\label{eq:ex}
J(O,T,\boldsymbol{E_k}) + J(\boldsymbol{E_k}) - J(O, \boldsymbol{E_k}) - J(T, \boldsymbol{E_k})
\end{align}

We argue that (\ref{eq:ex}) is minimized when the $k$-th attribute minimizes the Min-CIM and Min-Redundancy criteria.

As stated in \cite{peng2005feature}, the maximum of $J(O,\boldsymbol{E_k})$ is attained
when all variables are maximally dependent. When $O, \boldsymbol{E_{k-1}}$
are fixed, this indicates that the attribute $E_k$ should have the
maximal dependency to $O$. In this case, we get that $J(O,T,\boldsymbol{E_k}) = J(T,\boldsymbol{E_k})$. 
Note that when the dependency of $O$ or $T$ in $E_k$ increases, the conditional mutual information $I(O;T|\boldsymbol{E_k})$ decreases. This is the Min-CIM criterion.

Moreover, as noted in \cite{peng2005feature}, the minimum of $J(\boldsymbol{E_k})$ is attained when the attributes
$E_1,\ldots,E_k$ are
independent of each other. As all the attributes $E_1, \ldots, E_{k-1}$ are fixed at this point, this pair-wise independence condition
means that the mutual information between the attribute $E_k$ and any
other attribute $E_i$ is minimized. This is
the Min-Redundancy criterion. 

Thus, we get that the overall expression in (\ref{eq:ex}) is minimized (i.e., $J(O,\boldsymbol{E_k})$ is maximized, $J(O,\boldsymbol{E_k},T) = J(T,\boldsymbol{E_k})$, and $J(\boldsymbol{E_k})$ is minimized) when we are minimizing the Min-CIM and Min-Redundancy criteria.

\end{proof}

\begin{proof}[Proof of Lemma \ref{lemma:resposibility}]
First, since 
$(O\independent E_{k+1}|\boldsymbol{E_k})$ we have:\\ $I(O, E_{k+1} |\boldsymbol{E_k}) = 0$.
We get:
\begin{multline*}
 I(O;T|\boldsymbol{E_k}) -  I(O;T|\boldsymbol{E_k}, E_{k+1}) =\\
 H(O|\boldsymbol{E_k}) - H(O|T,\boldsymbol{E_k}) - H(O|\boldsymbol{E_k}, E_{k+1}) + H(O|T,\boldsymbol{E_k}, E_{k+1})
\end{multline*}
Since $H(O|\boldsymbol{E_k}) - H(O|E_{k+1},\boldsymbol{E_k}) = H(O|\boldsymbol{E_k}) - H(O|\boldsymbol{E_k}) =0$, we get:
\begin{multline*}
 I(O;T|\boldsymbol{E_k}) -  I(O;T|\boldsymbol{E_k}, E_{k+1}) =\\  H(O|T,\boldsymbol{E_k}, E_{k+1}) - H(O|T,\boldsymbol{E_k})  \leq 0
 \end{multline*}
For the last inequality we used the fact that for every three ransom variables $X,Y,X$: $H(X|Y) \leq H(X|Y,Z)$, since adding more conditions can only reduce the uncertainty of $X$. 

We get that the numerator of the responsibility score of $E_{k_1}$ is $\leq 0$, and thus $Resp(E_{k+1}) \leq 0$
\end{proof}

\begin{proposition}
The time complexity of the incremental MCIMR algorithm is $O(k |\mathcal{A}|)$.
\end{proposition}

\begin{proof}
At each iteration, the MCIMR algorithm selects a new attribute to be added based on the condition defined in Equation (5). In the worst case, it examines all attributes in $\mathcal{A}$. Since it stops after at most $k$ iterations, we get that the time complexity is $O(k |\mathcal{A}|)$.  
\end{proof}

We next prove that logical dependencies can lead to a misleading conclusion that we found a confounding attribute. 

\begin{lemma}
If for an attribute $E$ we have: $FD: E {\Rightarrow} T$ then we get $I(O;T|E,C) {=} 0$. 
\end{lemma}

\begin{proof}
If for an attribute $E$ we have: $FD: E {\Rightarrow} T$ then we have 
$H(T|E) {\approx} 0$. We get: 
$
 I(O;T|E,C) {=}$ $ H(O|E,C) {-} H(O|T,E,C)   
$.
But since $T$ and $E$ are dependent, we get: $H(O|T,E,C) {\approx} H(O|E,C)$ and thus $I(O;T|E,C) {=} 0$.
\end{proof}

The lemma also holds for the case that the attribute $E$ logically depends on the outcome $O$.

\textbf{Relevance Test}: Given a candidate attribute $E$, if $(O {\independent} E|C)$ and $(O {\independent} E|C,T)$ we get that $H(O|E,C) {=} H(O|C)$ and $H(O|T,E,C) {=} H(O|T,C)$. 
Thus: 
$$I(O;T|E,C) {=} H(O|E,C) {-} H(O|T,E,C) {=} H(O|C) {-} H(O|T,C) {=} I(O;T|C)$$
That means that the individual explanation power of $E$ is low, and thus it can be dropped as we assume $E$ cannot be a part of the optimal explanation. 

\section{Experiments}

\paragraph*{Explanation quality}
Next, we provide references supporting the explanations generated by \sysName. These in-domain findings serve as "domain-expert" explanations.

\textbf{SO Q1}: It was shown in \cite{stackoverflowreport} that there is a correlation between the developers salary and countries' economies. In \url{https://www.daxx.com/blog/development-trends/it-salaries-software-developer-trends}, it was also shown that the countries with the highest salary for developers are countries with a relatively high HDI (e.g., the US, Switzerland, Denmark). 

\textbf{SO Q2 + Q3}: It was mentioned in \url{https://content.techgig.com/career-advice/what-is-the-average-salary-of-software-engineers-in-different-countries/articleshow/91121900.cms} that countries that have a scarcity of software graduates tend to offer higher salaries than countries like India which produce hundred of thousands developers every year. This suggests that besides the economy of a country (resp., continent), the population size is also a factor that affects the average salary of developers.

\textbf{Flights Q1}: It was stated in \cite{usatoday} that weather is one of the top reasons for flights delay in the US.

\textbf{Flights Q2-Q4}: It was mentioned in \url{https://www.bts.gov/topics/airlines-and-airports/understanding-reporting-causes-flight-delays-and-cancellations} that besides weather conditions, main causes for flights delay in the US are heavy traffic volume, and air traffic control. Those two factors are highly correlated with population size. In bigger and more dense area, the air traffic increases. 

\textbf{Flights Q5}: It was mentioned in \url{https://www.bts.gov/topics/airlines-and-airports/understanding-reporting-causes-flight-delays-and-cancellations} that a main cause of the delay of flights in the US is the airline's control (e.g., maintenance or crew problems).  

\textbf{Covid Q1}: It was shown that there is a correlation between countries' economies and Covid-19 death rate \cite{upadhyay2021correlation,kaklauskas2022effects}.

\textbf{Covid Q2-Q3}: It was stated in \cite{pascoal2022population,martins2021relationship} that population density impact on COVID-19 mortality rate.

\textbf{Forbes Q1}: It was shown in \url{https://www.theguardian.com/world/2019/sep/15/hollywoods-gender-pay-gap-revealed-male-stars-earn-1m-more-per-film-than-women} that there is a gender pay gap for actors in Hollywood. Thus, it make sense that gender is a factor affecting the average salary of actors. It was also stated in \url{https://www.gobankingrates.com/money/jobs/how-much-do-actors-make/} that actors get paid according to their experience, which is reflected in their net worth.

\textbf{Forbes Q2}: It was mentioned in \url{https://climbtheladder.com/producer-salary/} that what affects directors and producers salary is their level of experience (which is reflected in the awards and net worth attributes). 

\textbf{Forbes Q3}: It was stated in \url{} that very often professional athletes salaries are performance-based. The performance quality is reflected in the Cups and Draft Pick attributes (for tennis, basketball and  football athletes, which are the majority of athletes in the Forbes dataset).

We next present additional experiments. 

\vspace{1mm}
\textit{Impact of pruning}.  
We next examine how useful were our pruning techniques. 
\textbf{Offline Pruning}. We found that our two offline pruning optimizations to be highly useful: On average, we dropped 41\%, 59\%, 45\%, and 73\% of the extracted attributes, in the SO, Flights, Covid-19, and Forbes dataset, resp. 
\textbf{Online Pruning}. 
At query time, we filter the extracted attributes using the logical dependency and the low relevance techniques.
Not surprisingly, as most irrelevant attributes were already dropped in the offline phase, we dropped many fewer attributes at this phase. On average, we dropped  14\%, 6\%, 11\% and 3\% of the remaining attributes, in SO, Flights, Covid-19, and Forbes, resp.

\vspace{1mm}
\textit{Entity linker}. Many of the missing values were caused by an unsuccessful matching of values from the table to their entities in the KG. For example, in SO, for some developers, their origin country is \texttt{Russian Federation}. However, the corresponding entity in DBpedia is called \texttt{Russia}. We thus failed to extract the properties of this country. In other cases, the values that appear in the tables were ambiguous, and thus we failed to match them to DBpedia entities. For example, in Forbes, one of the athletes is called \texttt{Ronaldo}. SpaCy entity linker could not decide whether to link this value to the entity \texttt{Ronaldo Luís Nazário de Lima} (Brazilian footballer) or to \texttt{Cristiano Ronaldo} (Portuguese footballer).

\vspace{1mm}
\textit{Multi-Hops}.
We examine the effect of extracting attributes following more than one hop in the KG. We report that in the vast majority of cases, \sysName's explanations were unaffected.
Namely, almost all attributes extracted from 2 or more hops were found to be irrelevant (and were pruned).
In some cases, we found at most one more attribute that was included in the explanations. For example, in Forbes $Q_1$, an attribute representing the average budget of the films played by actors (attribute extracted from 2-hops) was included in the explanation. In all cases, no attributes from 3 or more hops was considered to be relevant. 
Further, since the number of candidate attributes was increased (in 145\%, on average), running times were increased (by up to 15 seconds). This indicates that most of the relevant information can be found in the first hop. Future research will predict which paths in the KG may lead to relevant attributes.  
\end{document}